\documentclass[conference]{IEEEtran}
\usepackage[latin9]{inputenc}
\usepackage{float}
\usepackage{amsmath}
\usepackage{amssymb}
\usepackage{stmaryrd}
\usepackage{graphicx}
\usepackage{amsthm}
\usepackage{array}

\makeatletter

\floatstyle{ruled}
\newfloat{algorithm}{tbp}{loa}
\providecommand{\algorithmname}{Algorithm}
\floatname{algorithm}{\protect\algorithmname}

\IEEEoverridecommandlockouts
\usepackage{cite}
\usepackage{amsfonts}
\usepackage{textcomp}
\usepackage{xcolor}
\usepackage{epstopdf}

\def\BibTeX{{\rm B\kern-.05em{\sc i\kern-.025em b}\kern-.08em
    T\kern-.1667em\lower.7ex\hbox{E}\kern-.125emX}}

\theoremstyle{plain}
\newtheorem{theorem}{\protect\theoremname}
\theoremstyle{remark}

\theoremstyle{plain}

\theoremstyle{plain}

\newtheorem{remark}{Remark}

\newtheorem{property}{Property}

\usepackage[english]{babel}
\providecommand{\corollaryname}{Corollary}
\providecommand{\remarkname}{Remark}
\providecommand{\theoremname}{Theorem}
\providecommand{\lemmaname}{Lemma}

\makeatother

\begin{document}

\title{Tensor-based Multi-dimensional Wideband Channel Estimation for mmWave
Hybrid Cylindrical Arrays }

\author{Zhipeng Lin, \emph{Member, IEEE}, Tiejun Lv, \emph{Senior Member,
IEEE}, Wei Ni, \emph{Senior Member, IEEE},\\
J. Andrew Zhang, \emph{Senior Member, IEEE}, and Ren Ping Liu, \emph{Senior
Member, IEEE}

\thanks{This work was supported by the National Natural Science
Foundation of China (NSFC) under Grant 61671072.
\emph{(Corresponding author: Tiejun Lv.)}

Z. Lin and T. Lv are with the School of Information and Communication
Engineering, BUPT, Beijing, China (email: \{linlzp, lvtiejun\}@bupt.edu.cn).
Z. Lin is also with the School of Electrical and Data Engineering,
UTS, Sydney, Australia.

W. Ni is with the Data 61, CSIRO, Sydney, Australia (e-mail: Wei.Ni@data61.csiro.au).

J. A. Zhang and R. P. Liu are with the School of Electrical and Data
Engineering, UTS, Sydney, Australia (e-mail: \{Andrew.Zhang, RenPing.Liu\}@uts.edu.au).}}
\maketitle
\begin{abstract}
Channel estimation is challenging for hybrid millimeter wave (mmWave)
large-scale antenna arrays which are promising in 5G/B5G applications.
The challenges are associated with angular resolution losses resulting
from hybrid front-ends, beam squinting, and susceptibility to the
receiver noises. Based on tensor signal processing, this paper presents
a novel multi-dimensional approach to channel parameter estimation
with large-scale mmWave hybrid uniform circular cylindrical arrays
(UCyAs) which are compact in size and immune to mutual coupling but
known to suffer from infinite-dimensional array responses and intractability.
We design a new resolution-preserving hybrid beamformer and a low-complexity
beam squinting suppression method, and reveal the existence of shift-invariance
relations in the tensor models of received array signals at the UCyA.
Exploiting these relations, we propose a new tensor-based subspace
estimation algorithm to suppress the receiver noises in all dimensions
(time, frequency, and space). The algorithm can accurately estimate
the channel parameters from both coherent and incoherent signals.
Corroborated by the  Cram\'{e}r-Rao lower bound (CRLB), simulation
results show that the proposed algorithm is able to achieve substantially
higher estimation accuracy than existing matrix-based techniques,
with a comparable computational complexity.
\end{abstract}


\section{Introduction}

\global\long\def\figurename{Fig.}
 Massive hybrid antenna arrays can balance the hardware cost and complexity
of wideband millimeter wave (mmWave) transceivers in fifth generation
(5G) and beyond 5G (B5G) mobile communications \cite{Shahmansoori}.
Wideband mmWave hybrid circular arrays are particularly interesting
owing to their compact size, strong immunity to mutual coupling, and
inherently symmetric structure that enables 360-degree azimuth coverage
\cite{Circular1}. Channel parameter estimation for wideband mmWave
hybrid circular arrays is challenging, due to high-dimensional parameters,
large signal bandwidth, large signal propagation loss, and subsequent
susceptibility to noises \cite{TensorWidenband,ya2,Inter1}.

Existing channel parameter estimation algorithms (for the azimuth
and elevation angles, and the propagation delay of an incident signal)
have typically been matrix-based. By those matrix-based algorithms,
the relations between different dimensions (i.e., domains) of the
signal become obscure, because the received multi-dimensional (i.e.,
space, time and frequency) signals are stacked into two-dimensional
matrices \cite{TensorDecompositions2,tensor1HoSVD}. Moreover, typical
high-resolution matrix-based subspace estimation algorithms, such
as multiple signal classification (MUSIC) \cite{34} and estimation
of signal parameters via rotational invariance techniques (ESPRIT)
\cite{we2}, were designed for narrowband systems, where channel parameters
vary negligibly within the system band and are unaffected by an adverse
beam squinting effect \cite{wideband}.

Wideband signal-subspace methods (WSSMs) \cite{ICCMwhole,GC} have
been used to remove the frequency dependence of array steering vectors
and suppress the beam squinting effect, before applying (narrowband)
subspace estimation algorithms in wideband mmWave systems. Existing
incoherent WSSMs (IWSSMs) \cite{we11,narrow} decompose received signals
into multiple non-overlapping narrowbands, and estimate the parameters
independently at each narrowband. These methods \cite{we11,narrow}
do not utilize the high temporal resolution offered by wideband mmWave
systems. In \cite{ICCMearly,ICCMwhole}, coherent WSSMs (CWSSMs) map
the frequency-dependent array steering matrices to a reference frequency
by producing so-called focusing matrices. The generation of the focusing
matrices in these methods requires initial values, and the performance
of the methods is susceptible to the initial values. A variation of
CWSSM, named unitary constrained array manifold interpolation (UCAMI),
is proposed in \cite{UCAM,TensorWidenband}. It eliminates the need
for initial estimates and avoids focusing loss\footnote{Focusing loss refers to the ratio between the array signal-to-noise
ratios after and before focusing operations. Focusing loss can be
avoided by constraining that the focusing matrices are unitary \cite{ICCMwhole}.}. However, the focusing matrices of UCAMI are obtained by solving
multi-dimensional optimization problems. The dimension of the problems
is equal to the number of estimation parameters, and UCAMI is computationally
expensive. To overcome the beam squinting effect, an approximated
channel model is developed in \cite{re1} to quantize the angular
space, which would introduce errors and grid mismatches leading to
a degraded channel estimation accuracy. To circumvent the grid mismatch,
the algorithm developed in \cite{re1} repeatedly refines the angular
grid and applies compressive sensing to estimate parameters. As a
result, multiple iterative reweighted least squares problems need
to be solved.

\begin{figure*}
\begin{centering}
\includegraphics[width=14cm]{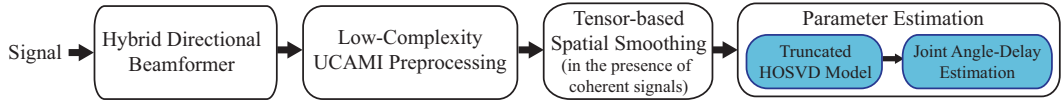}
\par\end{centering}
\centering{}\caption{The flow diagram of the proposed localization approach. From left
to right, the four key steps are described in Section III, Section
IV, Section V-C, and Sections V-A and V-B, respectively. \label{fig:The-flow-diagram}}
\end{figure*}

Tensor-based channel parameter estimations have been demonstrated
to be more powerful than conventional matrix-based techniques in \cite{TensorDecompositions1,TensorDecompositions2,tensorMagazine}.
By arranging the received signals in a tensor form and applying tensor
decomposition, the multi-dimensional parameters can be estimated with
super-high accuracy \cite{TensorDecompositions2,tensor1HoSVD}. The
papers \cite{tensor1HoSVD,tensorMUSIC,tensor3HoSVD,ICC} present tensor-based
algorithms for multi-dimensional channel parameter estimation, which
preserve the multi-dimensional structure of signals and improve estimation
accuracy in scatter-rich microwave-band channels. The authors of \cite{tensormmWave}
and \cite{tensorBFmmWave} exploit the sparsity of mmWave channels
to further improve the estimation accuracy. However, their algorithms
require an alternating-least-squares procedure with no guarantee of
convergence. In addition, the algorithm in \cite{tensorBFmmWave}
is only suitable for narrowband systems with uniform rectangular arrays
(URAs).

This paper presents a novel tensor-based approach for multi-dimensional
wideband channel estimation in large-scale mmWave hybrid uniform cylindrical
arrays (UCyAs). The key contributions of the paper are as follows:
\begin{itemize}
\item We design the hybrid beamformers by using quasi-discrete Fourier transform
(Q-DFT) to maintain the angular resolution of the hybrid UCyA with
a reduced number of radio frequency (RF) chains. Developing and applying
a low-complexity UCAMI, we suppress the beam squinting effect and
enable coherent combining across the wideband. These are two salient
steps for our new tensor-based parameter estimation.
\item We propose a new tensor-based subspace estimation algorithm to jointly
estimate the delay and the azimuth and elevation angles of each received
signal by exploiting the important shift-invariance relations in the
constructed truncated higher-order singular value decomposition (HOSVD)
model. The algorithm can suppress the receiver noises in all of the
time, frequency, and space dimensions, and hence accurately estimate
the high-dimensional channel parameters of multiple coherent or incoherent
signal sources.
\item We introduce a new way to rearrange the measurement tensor of the
received signals to decorrelate coherent signals at the hybrid UCyA,
i.e., spatial smoothing. Coherent signals can then be separated and
can be estimated independently by using the proposed tensor subspace
estimation algorithm.
\end{itemize}

The steps of the proposed approach are illustrated in Fig. \ref{fig:The-flow-diagram},
and elaborated on in the rest of this paper. In the first step, the
received signals are first synthesized by a hybrid directional beamformer,
which uses Q-DFT to reduce the number of required RF chains (with
a negligible cost of the channel estimation accuracy at the later
stages of the technique). In other words, this step reduces the dimension
of the received signals, so that the signals can be processed with
much fewer RF chains (than antennas). The second step is a proposed
low-complexity UCAMI, which suppresses the beam squinting effect efficiently
by only optimizing the focusing matrices in the elevation angular
domain. The third step is to reveal and exploit the inherent linear
recurrence relations in the first mode of the measure tensor and run
spatial smoothing to decorrelate the coherent signals. Finally, the
new tensor-based joint delay-angle estimation algorithm is carried
out to estimate the delay and azimuth and elevation angles based on
the constructed truncated HOSVD model of the measure tensor.

Validated by the  Cram\'{e}r-Rao lower bound (CRLB), simulation results
show that the proposed algorithm is able to achieve much higher accuracy
than state-of-the-art matrix-based techniques for wideband mmWave
hybrid UCyAs. The new tensor-based algorithms work well even when
the signal-to-noise ratio (SNR) is low, credited to the effective
noise suppression in all of the time, space, and frequency domains.

Different from the existing studies, e.g, \cite{re1}, we develop
a new low-complexity UCAMI to suppress the beam squinting effect,
which does not quantize the angular space{} and hence no quantization
error will occur. Moreover, we reveal and exploit inherent shift-invariance
relations \cite{fmri} in each domain/mode of the measurement tensor.
As a result, our algorithm only needs to solve a one-time HOSVD of
the measurement tensor to estimate the multi-dimensional parameters
jointly.

The rest of this paper is organized as follows. The system model is
introduced in Sections II. In Sections III and IV, we design the hybrid
beamformers and suppress the beam squinting effect in the received
signals. In Section V, we introduce the new tensor-based parameter
estimation algorithm. Simulations are provided in Section VI, followed
by conclusions in Section VII.

\subsection{Preliminary and Notation }

Notations $a$, $\mathbf{a}$, $\mathbf{A}$, and $\mathbb{A}$ stand
for scalar, column vector, matrix, and set, respectively. $\mathbf{I}_{K}$
and $\mathbf{0}_{M\times K}$ denote a $K\times K$ identity matrix
and an $M\times K$ zero matrix, respectively. $\mathbf{A}^{\ast}$,
$\mathbf{A}^{T}$ and $\mathbf{A}^{H}$ denote the conjugate, transpose
and conjugate transpose of $\mathbf{A}$, respectively. $\left\Vert \mathbf{A}\right\Vert _{\textrm{F}}$
denotes the Frobenius norm of $\mathbf{A}$. $\otimes$ and $\diamond$
denote the Kronecker product and Khatri-Rao product, respectively.
$\hat{a}$ denotes the estimate of $a$.

Tensor is the generalization of scalar (which has a zero-order mode),
vector (which has one-order mode), and matrix (which has two-order
modes) to arrays with an arbitrary order of modes. We use $\mathcal{A}\in\mathbb{C}^{I_{1}\times I_{2}\times\cdots\times I_{N}}$
to denote an order-$N$ tensor, whose elements (entries) are $a_{i_{1},i_{2},\cdots,i_{N}},$
$i_{n}=1,2,\ldots,I_{n}$, and the index of $\mathcal{A}$ in the
$n$-th mode ranges from 1 to $I_{n}$. By fixing some of the indices,
a subtensor of $\mathcal{A}$ can be formed: $\mathcal{A}_{:,:,\cdots,:,i_{n}=k,:,\cdots,:}$
with the index of the mode-$n$ set to $k$ $\left(0\leq k\leq I_{n}\right)$.
$\times_{n}$ and $\circ$ stand for tensor $n$-mode product and
outer product, respectively. $\left[\mathcal{A}\sqcup_{n}\mathcal{B}\right]$
denotes the tensor concatenation of $\mathcal{A}$ and $\mathcal{B}$
in mode-$n$. The mode-$n$ unfolding (also known as matricization)
of a tensor $\mathcal{A}\in\mathbb{C}^{I_{1}\times I_{2}\times\cdots\times I_{N}}$,
denoted by $\mathbf{A}_{(n)}\in\mathbb{C}^{I_{n}\times(I_{1}I_{2}\cdots I_{N}/I_{n})}$,
arranges the fibers in the $n$-th mode of $\mathcal{A}$ as the columns
of the resulting matrix $\mathbf{A}_{(n)}$. Some important properties
of tensor operations used in this paper are presented in Appendix
I.
\begin{figure*}
\begin{centering}
\includegraphics[width=14cm]{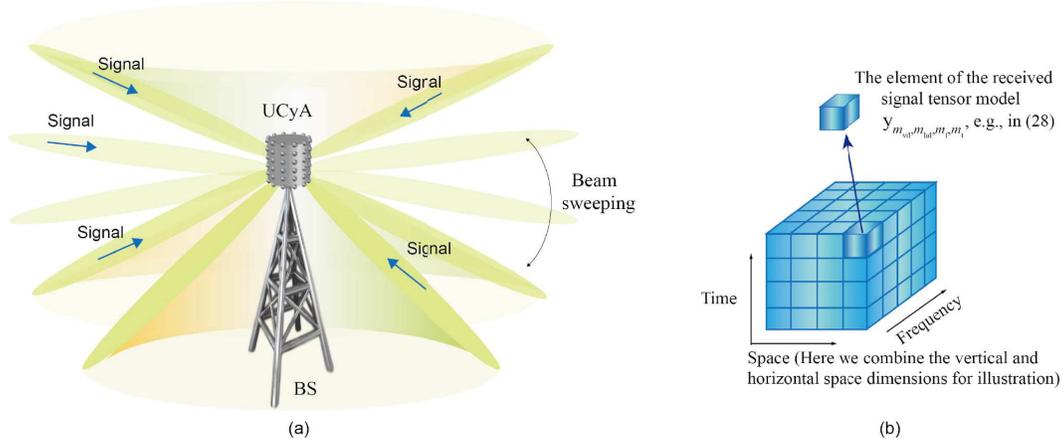}
\par\end{centering}
\centering{}\caption{Illustration on the proposed system and signal models. (a) System
configuration; (b) Signal tensor model.\label{fig:(a)-System-configuration;}}
\end{figure*}

\section{System Model}

In our system, a base station (BS) is equipped with a large-scale
hybrid mmWave UCyA with $M_{\textrm{bs}}$ antennas, consisting of
$M_{\textrm{v}}$ vertically placed uniform circular arrays (UCAs)
each with $M_{\textrm{h}}$ antenna elements, and $M_{\textrm{bs}}=M_{\textrm{v}}M_{\textrm{h}}$.
Let $r$ be the radius of the UCyA, and $h$ be the vertical distance
between any two adjacent vertical elements. A hybrid front-end is
adopted (i.e., there are fewer RF chains than antennas) with consideration
of hardware cost, energy consumption, weight and size. Consider a
wideband orthogonal frequency division multiplexing (OFDM) system,
with $M_{\textrm{f}}$ subcarriers. There are a total of $K$ three-dimensional
(3D) sources, each of which is equipped with a single antenna with
an isotropic beam pattern.

We apply vertical beam sweeping to obtain the signals from the sources,
as shown in Fig.\ref{fig:(a)-System-configuration;}(a). $M_{\textrm{b}}$
evenly spaced elevation angles are swept successively. For each elevation
angle, signal samples of $M_{\textrm{t}}$ time frames are collected
within a sweeping time interval $\tau_{\textrm{b}}$. In the $m_{\textrm{b}}$-th
sweeping beam $(m_{\textrm{b}}=1,\ldots,M_{\textrm{b}})$, the signals
from $K_{m_{\textrm{b}}}$ sources are captured at the BS (and $K\leq\sum_{m_{\textrm{b}}=1}^{M_{\textrm{b}}}K_{m_{\textrm{b}}}$,
due to the partially overlapping sweeping beams). The signal sample
associated with the $m_{\textrm{f}}$-th subcarrier $(m_{\textrm{f}}=1,\ldots,M_{\textrm{f}})$
at the $m_{\textrm{t}}$-th time frame $(m_{\textrm{t}}=1,\ldots,M_{\textrm{t}})$
can be expressed as \cite{Shahmansoori}:
\begin{align}
\mathbf{x}_{m_{\textrm{f}},m_{\textrm{t}},m_{\textrm{b}}} & =\sum_{^{k_{m_{\textrm{b}}}=1}}^{K_{m_{\textrm{b}}}}s_{m_{\textrm{t}},k_{m_{\textrm{b}}}}a_{\textrm{fbs},m_{\textrm{f}},m_{\textrm{b}}}(\tau_{k_{m_{\textrm{b}}}})\mathbf{B}_{m_{\textrm{f}},m_{\textrm{b}}}^{H}\nonumber \\
 & \times\mathbf{a}_{\textrm{bs},m_{\textrm{f}},m_{\textrm{b}}}(\phi_{k_{m_{\textrm{b}}}},\theta_{k_{m_{\textrm{b}}}})+\mathbf{n}_{m_{\textrm{f}},m_{\textrm{t}},m_{\textrm{b}}},\label{eq:ori_signalmodel}
\end{align}
where $\phi_{k_{m_{\textrm{b}}}}$ and $\theta_{k_{m_{\textrm{b}}}}$
are the azimuth and elevation angles-of-arrivals (AOAs) of the $k_{m_{\textrm{b}}}$-th
path, respectively; $\mathbf{a}_{\textrm{bs},m_{\textrm{f}},m_{\textrm{b}}}(\phi_{k_{m_{\textrm{b}}}},\theta_{k_{m_{\textrm{b}}}})\in\mathbb{C}^{M_{\textrm{bs}}}$
denotes the steering vector of the hybrid UCyA; $s_{m_{\textrm{t}},k_{m_{\textrm{b}}}}=\alpha_{k_{m_{\textrm{b}}}}\tilde{s}_{m_{\textrm{t}},k_{m_{\textrm{b}}}}/\sqrt{\rho_{k_{m_{\textrm{b}}}}}$,
where $\tilde{s}_{m_{\textrm{t}},k_{m_{\textrm{b}}}}$ is the transmitted
symbol, $\alpha_{k_{m_{\textrm{b}}}}$ is the signal power, and $\rho_{k_{m_{\textrm{b}}}}$
is the pathloss from the $k_{m_{\textrm{b}}}$-th source to the BS;
$\mathbf{n}_{m_{\textrm{f}},m_{\textrm{b}},m_{\textrm{t}}}\in\in\mathbb{C}^{M_{\textrm{bsd}}}$
denotes the additive white Gaussian noise (AWGN); $\mathbf{B}_{m_{\textrm{f}},m_{\textrm{b}}}=\mathbf{B}_{\textrm{ab}}\mathbf{B}_{\textrm{db},m_{\textrm{f}},m_{\textrm{b}}}\in\mathbb{C}^{M_{\textrm{bs}}\times M_{\textrm{bsd}}}$
is the hybrid beamforming matrix, composed of an analog beamforming
matrix $\mathbf{B}_{\textrm{ab}}\in\mathbb{C}^{M_{\textrm{bs}}\times M_{\textrm{bsr}}}$
and a digital beamforming matrix $\mathbf{B}_{\textrm{db},m_{\textrm{f}},m_{\textrm{b}}}\in\mathbb{C}^{M_{\textrm{bsr}}\times M_{\textrm{bsd}}}$;
$M_{\textrm{bsr}}$ is the number of RF chains; $M_{\textrm{bsd}}$
is the number of data streams after hybrid beamforming; and
\begin{equation}
a_{\textrm{fbs},m_{\textrm{f}},m_{\textrm{b}}}(\tau_{k_{m_{\textrm{b}}}})=a_{\textrm{f},m_{\textrm{f}}}(\tau_{k_{m_{\textrm{b}}}})b_{\textrm{f},m_{\textrm{b}}},\label{eq:Af}
\end{equation}
where $a_{\textrm{f},m_{\textrm{f}}}(\tau_{k_{m_{\textrm{b}}}})=e^{-j2\pi f_{m_{\textrm{f}}}\tau_{k_{m_{\textrm{b}}}}}$
and $b_{\textrm{f},m_{\textrm{b}}}=e^{-j2\pi f_{m_{\textrm{f}}}(m_{\textrm{b}}-1)\tau_{\textrm{b}}}$
with $\tau_{k_{m_{\textrm{b}}}}$ being the delay of the $k_{m_{\textrm{b}}}$-th
signal and $f_{m_{\textrm{f}}}$ being the $m_{\textrm{f}}$-th subcarrier
frequency. The delay $\tau_{k_{m_{\textrm{b}}}}$ can be used to estimate
the source distance.

Given the structure of UCyA, the array steering vector, i.e., $\mathbf{a}_{\textrm{bs},m_{\textrm{f}},m_{\textrm{b}}}(\phi_{k_{m_{\textrm{b}}}},\theta_{k_{m_{\textrm{b}}}})$,
can be given as the Kronecker product of the vertical and horizontal
array steering vectors:
\begin{align}
 & \mathbf{a}_{\textrm{bs},m_{\textrm{f}},m_{\textrm{b}}}(\phi_{k_{m_{\textrm{b}}}},\theta_{k_{m_{\textrm{b}}}})\nonumber \\
 & =\mathbf{a}_{\textrm{v},m_{\textrm{f}},m_{\textrm{b}}}(\theta_{k_{m_{\textrm{b}}}})\otimes\mathbf{a}_{\textrm{h},m_{\textrm{f}},m_{\textrm{b}}}(\theta_{k_{m_{\textrm{b}}}},\phi_{k_{m_{\textrm{b}}}}).\label{eq:Abs}
\end{align}
The elements of $\mathbf{a}_{\textrm{v},m_{\textrm{f}},m_{\textrm{b}}}(\theta_{k_{m_{\textrm{b}}}})$
and $\mathbf{a}_{\textrm{h},m_{\textrm{f}},m_{\textrm{b}}}(\theta_{k_{m_{\textrm{b}}}},\phi_{k_{m_{\textrm{b}}}})$
are:
\begin{align}
\left[\mathbf{a}_{\textrm{v},m_{\textrm{f}},m_{\textrm{b}}}(\theta_{k_{m_{\textrm{b}}}})\right]_{m_{\textrm{v}},1}=a_{\textrm{v},m_{\textrm{v}},m_{\textrm{f}},m_{\textrm{b}}}(\theta_{k_{m_{\textrm{b}}}})\nonumber \\
=\frac{1}{\sqrt{M_{\textrm{v}}}}\exp\left(-j\frac{2\pi}{c}f_{m_{\textrm{f}}}h(m_{\textrm{v}}-1)\cos(\theta_{k_{m_{\textrm{b}}}})\right),\label{eq:verticalarray}
\end{align}
\begin{align}
 & \left[\mathbf{a}_{\textrm{h},m_{\textrm{f}},m_{\textrm{b}}}(\theta_{k_{m_{\textrm{b}}}},\phi_{k_{m_{\textrm{b}}}})\right]_{m_{\textrm{h}},1}=a_{\textrm{h},m_{\textrm{h}},m_{\textrm{f}},m_{\textrm{b}}}(\theta_{k_{m_{\textrm{b}}}},\phi_{k_{m_{\textrm{b}}}})\nonumber \\
 & =\frac{1}{\sqrt{M_{\textrm{h}}}}\exp\left(j\frac{2\pi}{c}f_{m_{\textrm{f}}}r\sin(\theta_{k_{m_{\textrm{b}}}})\cos(\phi_{k_{m_{\textrm{b}}}}-\varphi_{m_{\textrm{h}}})\right),\label{eq:horizonarray}
\end{align}
where $c$ is the speed of light, and $\varphi_{m_{\textrm{h}}}=2\pi(m_{\textrm{h}}-1)/M_{\textrm{h}}$
is the difference between the central angles of the $m_{\textrm{h}}$-th
antenna and the first antenna of each UCA.

\section{Hybrid Directional Beamforming Design}

In this section, we design the analog and digital beamforming matrices,
$\mathbf{B}_{\textrm{ab}}$ and $\mathbf{B}_{\textrm{db},m_{\textrm{f}},m_{\textrm{b}}}$,
for the hybrid directional beamformer, as the first step shown in
Fig. 1. The number of required RF chains is reduced while the angular
resolution of the UCyA is not compromised as compared to its fully
digital counterparts.

We decouple $\mathbf{B}_{\textrm{ab}}$ between the vertical and horizontal
planes, i.e., $\mathbf{B}_{\textrm{ab}}=\mathbf{B}_{\textrm{vab}}\otimes\mathbf{B}_{\textrm{hab}}$
with $\mathbf{B}_{\textrm{vab}}\in\mathbb{C}^{M_{\textrm{v}}\times M_{\textrm{vr}}}$
and $\mathbf{B}_{\textrm{hab}}\in\mathbb{C}^{M_{\textrm{h}}\times M_{\textrm{hr}}}$.
By decoupling the beamformers into the Kronecker products of horizontal
and vertical matrices, we preserve the shift-invariance relations
on the vertical and horizontal planes, as will be revealed later in
Section V. To maintain the angular resolution of the hybrid UCyA,
we design $\mathbf{B}_{\textrm{hab}}$ based on the following theorem.

\begin{theorem} Suppose that $M_{\textrm{h}}\geq\left\lfloor 4\pi f_{m_{\textrm{f}}}r/c\right\rfloor $.
The array response vector $\mathbf{a}_{\textrm{h},m_{\textrm{f}},m_{\textrm{b}}}(\theta_{k_{m_{\textrm{b}}}},\phi_{k_{m_{\textrm{b}}}})$
can be transformed into a beamspace by using Q-DFT. If the index for
a beamspace dimension, $p$, is larger than $\left\lfloor 2\pi f_{m_{\textrm{f}}}r/c\right\rfloor $,
the element in the dimension is negligible and can be suppressed.
The expression for the elements in the other dimensions is given by:
\textit{
\begin{align}
 & a_{\textrm{QDFT},p,m_{\textrm{f}},m_{\textrm{b}}}(\theta_{k_{m_{\textrm{b}}}},\phi_{k_{m_{\textrm{b}}}})\nonumber \\
 & \approx\sqrt{M_{\textrm{h}}}j^{p}J_{p}\left(\gamma_{m_{\textrm{f}}}(\theta_{k_{m_{\textrm{b}}}})\right)\exp\left(-jp\phi_{k_{m_{\textrm{b}}}}\right),\label{eq:Theorem1}
\end{align}
}where $\gamma_{m_{\textrm{f}}}(\theta_{k_{m_{\textrm{b}}}})=2\pi f_{m_{\textrm{f}}}r\sin(\theta_{k_{m_{\textrm{b}}}})/c$,
$p=-P,-P+1,\ldots,P$, and $J_{p}\left(\gamma_{m_{\textrm{f}}}(\theta_{k_{m_{\textrm{b}}}})\right)$
is the Bessel function of the first kind of order $p$. \end{theorem}
\begin{proof} See Appendix II. \end{proof}

Theorem 1 shows that, with the application of Q-DFT \cite{book11},
the $M_{\textrm{h}}$-dimensional array response vector of each UCA,
$\mathbf{a}_{\textrm{h},m_{\textrm{f}},m_{\textrm{b}}}(\theta_{k_{m_{\textrm{b}}}},\phi_{k_{m_{\textrm{b}}}})$,
can be transformed to be $(2P+1)$-dimensional, where $P=\left\lfloor 2\pi f_{m_{\textrm{f}}}r/c\right\rfloor $.
As a result, only $M_{\textrm{hr}}=(2P+1)$ RF chains are required
on the horizontal plane. Specifically, according to Theorem 1, we
design $\mathbf{B}_{\textrm{hab}}$ as $\left[\mathbf{B}_{\textrm{hab}}\right]_{m_{\textrm{h}},m_{\textrm{hr}}+P+1}=e^{-j2\pi(m_{\textrm{h}}-1)m_{\textrm{hr}}/M_{\textrm{h}}},$
where $m_{\textrm{hr}}=-P,-P+1,\ldots,P$. We set $\mathbf{B}_{\textrm{vab}}=\mathbf{I}_{M_{\textrm{v}}}$
to preserve the recurrence relation between the UCAs, i.e., the shift-invariance
relation. The relation is crucial for the subspace-based estimation
algorithms, and exploited to estimate the elevation AOAs in this paper.
With this design, the number of required RF chains is only $M_{\textrm{bsr}}=M_{\textrm{vr}}M_{\textrm{hr}}=M_{\textrm{v}}(2P+1)$.

Then, we design the digital beamformer $\mathbf{B}_{\textrm{db},m_{\textrm{f}},m_{\textrm{b}}}$
as
\begin{align}
\mathbf{B}_{\textrm{db},m_{\textrm{f}},m_{\textrm{b}}}=\textrm{diag}\left(b_{\textrm{db},1,m_{\textrm{f}},m_{\textrm{b}}},\ldots,b_{\textrm{db},M_{\textrm{bsr}},m_{\textrm{f}},m_{\textrm{b}}}\right),\label{eq:digitalBF}
\end{align}
where $b_{\textrm{db},m_{\textrm{bsr}},m_{\textrm{f}},m_{\textrm{b}}}$
$(m_{\textrm{bsr}}=1,2,\ldots,M_{\textrm{bsr}})$ is the beamforming
weight coefficients. Since $\mathbf{B}_{\textrm{db},m_{\textrm{f}},m_{\textrm{b}}}$
is diagonal, we have $M_{\textrm{bsr}}=M_{\textrm{bsd}}$. Considering
that sweeping beams on both the vertical and horizontal planes would
take a longer time, we design the beamformers to sweep on the vertical
plane only, and operate omnidirectionally on the horizon plane. The
beamforming weight coefficients can be configured according to the
beamforming response, $P_{m_{\textrm{f}}}(\bar{\theta}_{m_{\textrm{b}}}),$
as given by
\begin{align}
P_{m_{\textrm{f}}}(\bar{\theta}_{m_{\textrm{b}}}) & =\mathbf{b}_{\textrm{db},m_{\textrm{f}},m_{\textrm{b}}}^{H}\mathbf{B}_{\textrm{ab}}^{H}\mathbf{a}_{\textrm{bs},m_{\textrm{f}},m_{\textrm{b}}}(\bar{\theta}_{m_{\textrm{b}}},\phi),
\end{align}
where
\begin{equation}
\mathbf{b}_{\textrm{db},m_{\textrm{f}},m_{\textrm{b}}}=\left[b_{\textrm{db},1,m_{\textrm{f}},m_{\textrm{b}}},\ldots,b_{\textrm{db},M_{\textrm{bsr}},m_{\textrm{f}},m_{\textrm{b}}}\right]^{T}
\end{equation}
is the normalized digital beamforming vector, i.e., $\mathbf{b}_{\textrm{db},m_{\textrm{f}},m_{\textrm{b}}}^{H}\mathbf{b}_{\textrm{db},m_{\textrm{f}},m_{\textrm{b}}}=1$,
and $\bar{\theta}_{m_{\textrm{b}}}$ is the $m_{\textrm{b}}$-th beamforming
sweeping direction. Assume that the vertical angular sweeping interval
is $\frac{\pi}{M_{\textrm{b}}}$. The elevation angle of the $m_{\textrm{b}}$-th
angular sample ranges from $\frac{\pi}{M_{\textrm{b}}}(m_{\textrm{b}}-1)$
to $\frac{\pi}{M_{\textrm{b}}}m_{\textrm{b}}$.

We also decouple the digital beamforming matrix $\mathbf{B}_{\textrm{db},m_{\textrm{f}},m_{\textrm{b}}}$
in \eqref{eq:digitalBF} between the vertical and horizontal planes,
i.e., $\mathbf{B}_{\textrm{db},m_{\textrm{f}},m_{\textrm{b}}}=\mathbf{B}_{\textrm{vdb},m_{\textrm{f}},m_{\textrm{b}}}\otimes\mathbf{B}_{\textrm{hdb},m_{\textrm{f}},m_{\textrm{b}}}$,
where $\mathbf{B}_{\textrm{vdb},m_{\textrm{f}},m_{\textrm{b}}}\in\mathbb{C}^{M_{\textrm{vd}}\times M_{\textrm{vd}}}$
and $\mathbf{B}_{\textrm{hdb},m_{\textrm{f}},m_{\textrm{b}}}\in\mathbb{C}^{M_{\textrm{hd}}\times M_{\textrm{hd}}}$
are diagonal matrices with elements $b_{\textrm{vdb},m_{\textrm{vd}},m_{\textrm{f}},m_{\textrm{b}}}$
and $b_{\textrm{hdb},m_{\textrm{hd}},m_{\textrm{f}},m_{\textrm{b}}}$,
respectively. Thus, after hybrid beamforming, the array steering vectors
$\mathbf{a}_{\textrm{bs},m_{\textrm{f}},m_{\textrm{b}}}(\theta_{k_{m_{\textrm{b}}}},\phi_{k_{m_{\textrm{b}}}})$
can be written as:
\begin{align}
 & \mathbf{a}_{\textrm{hd},m_{\textrm{f}},m_{\textrm{b}}}(\theta_{k_{m_{\textrm{b}}}},\phi_{k_{m_{\textrm{b}}}})=\mathbf{B}_{m_{\textrm{f}},m_{\textrm{b}}}^{H}\mathbf{a}_{\textrm{bs},m_{\textrm{f}},m_{\textrm{b}}}(\theta_{k_{m_{\textrm{b}}}},\phi_{k_{m_{\textrm{b}}}})\nonumber \\
 & =\left(\left(\mathbf{B}_{\textrm{vab}}\otimes\mathbf{B}_{\textrm{hab}}\right)\left(\mathbf{B}_{\textrm{vdb},m_{\textrm{f}},m_{\textrm{b}}}\otimes\mathbf{B}_{\textrm{hdb},m_{\textrm{f}},m_{\textrm{b}}}\right)\right)^{H}\nonumber \\
 & \qquad\times\mathbf{a}_{\textrm{bs},m_{\textrm{f}},m_{\textrm{b}}}(\theta_{k_{m_{\textrm{b}}}},\phi_{k_{m_{\textrm{b}}}})\nonumber \\
 & \stackrel{(\textrm{a})}{=}\left(\left(\mathbf{B}_{\textrm{vab}}\mathbf{B}_{\textrm{vdb},m_{\textrm{f}},m_{\textrm{b}}}\right)\otimes\left(\mathbf{B}_{\textrm{hab}}\mathbf{B}_{\textrm{hdb},m_{\textrm{f}},m_{\textrm{b}}}\right)\right)^{H}\nonumber \\
 & \qquad\times\mathbf{a}_{\textrm{bs},m_{\textrm{f}},m_{\textrm{b}}}(\theta_{k_{m_{\textrm{b}}}},\phi_{k_{m_{\textrm{b}}}})\nonumber \\
 & \stackrel{(\textrm{b})}{=}\left(\left(\mathbf{B}_{\textrm{vab}}\mathbf{B}_{\textrm{vdb},m_{\textrm{f}},m_{\textrm{b}}}\right)^{H}\otimes\left(\mathbf{B}_{\textrm{hab}}\mathbf{B}_{\textrm{hdb},m_{\textrm{f}},m_{\textrm{b}}}\right)^{H}\right)\nonumber \\
 & \qquad\times\left(\mathbf{a}_{\textrm{v},m_{\textrm{f}},m_{\textrm{b}}}(\theta_{k_{m_{\textrm{b}}}})\otimes\mathbf{a}_{\textrm{h},m_{\textrm{f}},m_{\textrm{b}}}(\theta_{k_{m_{\textrm{b}}}},\phi_{k_{m_{\textrm{b}}}})\right)\nonumber \\
 & =\mathbf{a}_{\textrm{vhb},m_{\textrm{f}},m_{\textrm{b}}}(\theta_{k_{m_{\textrm{b}}}})\otimes\mathbf{a}_{\textrm{hhb},m_{\textrm{f}},m_{\textrm{b}}}(\theta_{k_{m_{\textrm{b}}}},\phi_{k_{m_{\textrm{b}}}}),\label{eq:afterBF}
\end{align}
where $\mathbf{a}_{\textrm{vhb},m_{\textrm{f}},m_{\textrm{b}}}(\theta_{k_{m_{\textrm{b}}}})\in\mathbb{C}^{M_{\textrm{vd}}}$,
$\mathbf{a}_{\textrm{hhb},m_{\textrm{f}},m_{\textrm{b}}}(\theta_{k_{m_{\textrm{b}}}},\phi_{k_{m_{\textrm{b}}}})\in\mathbb{C}^{M_{\textrm{hd}}}$,
$M_{\textrm{vd}}=M_{\textrm{vr}}=M_{\textrm{v}}$, and $M_{\textrm{hd}}=M_{\textrm{hr}}=2P+1$.
In \eqref{eq:afterBF}, $(\textrm{a})$ and $(\textrm{b})$ are based
on two important properties of the Kronecker product, i.e., $(\mathbf{A}\otimes\mathbf{B})(\mathbf{C}\otimes\mathbf{D})=\mathbf{AC}\otimes\mathbf{BD}$
and $(\mathbf{A}\otimes\mathbf{B})^{H}=\mathbf{A}^{H}\otimes\mathbf{B}^{H}$
\cite{Arraysignalprocessingconceptsandtechniques}. We have
\begin{align}
\mathbf{a}_{\textrm{vhb},m_{\textrm{f}},m_{\textrm{b}}}(\theta_{k_{m_{\textrm{b}}}}) & =\mathbf{B}_{\textrm{vdb},m_{\textrm{f}},m_{\textrm{b}}}^{H}\mathbf{a}_{\textrm{v},m_{\textrm{f}},m_{\textrm{b}}}(\theta_{k_{m_{\textrm{b}}}}).\label{eq:Avhb}
\end{align}
According to Theorem 1, the $m_{\textrm{hd}}$-th element of $\mathbf{a}_{\textrm{hhb},m_{\textrm{f}},m_{\textrm{b}}}(\theta_{k_{m_{\textrm{b}}}},\phi_{k_{m_{\textrm{b}}}})$
is given by:
\begin{align}
 & a_{\textrm{hhb},m_{\textrm{hd}},m_{\textrm{f}},m_{\textrm{b}}}(\theta_{k_{m_{\textrm{b}}}},\phi_{k_{m_{\textrm{b}}}})\approx\sqrt{M_{\textrm{h}}}j^{m_{\textrm{hd}}}b_{\textrm{hdb},m_{\textrm{hd}},m_{\textrm{f}},m_{\textrm{b}}}\nonumber \\
 & \qquad\times J_{m_{\textrm{hd}}}\left(\gamma_{m_{\textrm{f}}}(\theta_{k_{m_{\textrm{b}}}})\right)\exp\left(-jm_{\textrm{hd}}\phi_{k_{m_{\textrm{b}}}}\right).\label{eq:Ahhb}
\end{align}

Given our hybrid beamforming design, we can present the beamspace
signals of the mmWave UCyA in a tensor form. Considering the observations
at all sweeping intervals, subcarriers and time frames, the beamspace
signals can be modeled as:
\begin{align}
 & x_{m_{\textrm{vd}},m_{\textrm{hd}},m_{\textrm{f}},m_{\textrm{t}},m_{\textrm{b}}}=\sum_{^{k_{m_{\textrm{b}}}=1}}^{K_{m_{\textrm{b}}}}\left(s_{m_{\textrm{t}},k_{m_{\textrm{b}}}}a_{\textrm{vhb},m_{\textrm{vd}},m_{\textrm{f}},m_{\textrm{b}}}(\theta_{k_{m_{\textrm{b}}}})\right.\nonumber \\
 & \left.\times a_{\textrm{hhb},m_{\textrm{hd}},m_{\textrm{f}},m_{\textrm{b}}}(\theta_{k_{m_{\textrm{b}}}},\phi_{k_{m_{\textrm{b}}}})a_{\textrm{fbs},m_{\textrm{f}},m_{\textrm{b}}}(\tau_{k_{m_{\textrm{b}}}})\right)\nonumber \\
 & +n_{m_{\textrm{vd}},m_{\textrm{hd}},m_{\textrm{f}},m_{\textrm{t}},m_{\textrm{b}}},\label{eq:measurementsample}
\end{align}
where $n_{m_{\textrm{vd}},m_{\textrm{hd}},m_{\textrm{f}},m_{\textrm{t}},m_{\textrm{b}}}$
is the additive noise.

We consider the samples from the $m_{\textrm{b}}$-th vertical sweeping
beam, and \eqref{eq:measurementsample} can be rewritten in the following
tensor form \cite{tensor1HoSVD}
\begin{equation}
\mathcal{X}_{:,:,:,:,m_{\textrm{b}}}=\mathcal{A}_{m_{\textrm{b}}}\times_{4}\mathbf{S}_{m_{\textrm{b}}}+\mathcal{N}_{m_{\textrm{b}}}\in\mathbb{C}^{M_{\textrm{vd}}\times M_{\textrm{hd}}\times M{}_{\textrm{f}}\times M_{\textrm{t}}},\label{eq:tensormodel}
\end{equation}
where all the angle and delay parameters at the $m_{\textrm{b}}$-th
sweeping beam are collected in the space-time response tensor $\mathcal{A}_{m_{\textrm{b}}}\in\mathbb{C}^{M_{\textrm{vd}}\times M_{\textrm{hd}}\times M{}_{\textrm{f}}\times K_{m_{\textrm{b}}}}$;
$\mathbf{S}_{m_{\textrm{b}}}\in\mathbb{C}^{M_{\textrm{t}}\times K_{m_{\textrm{b}}}}$
collects the received symbols $s_{m_{\textrm{t}},k_{m_{\textrm{b}}}}$;
and $\mathcal{N}_{m_{\textrm{b}}}\in\mathbb{C}^{M_{\textrm{vd}}\times M_{\textrm{hd}}\times M{}_{\textrm{f}}\times M_{\textrm{t}}}$
collects the noise samples.

\section{Low-Complexity Coherent Preprocessing for Wideband Signals}

As the second step in Fig. 1, a new low-complexity UCAMI is developed
in this section to suppress the beam squinting effect and enable coherent
combining of measurement signals across wideband. The conventional
UCAMI \cite{UCAM} needs to solve a computationally expensive multi-dimensional
optimization problem whose dimension is equal to the number of estimation
parameters. Different from the conventional UCAMI, there are only
1-D problems in our proposed approach.

As shown in \eqref{eq:Avhb} and \eqref{eq:Ahhb}, the array steering
vectors depend on the frequency and so do the beamspace signals. As
a consequence, the signals can suffer from the beam squinting effect,
due to the wide bandwidth of mmWave signals. It is critical to preprocess
the beamspace signals in order to suppress the frequency dependence
of the array steering vectors. The suppression of frequency dependence
is performed by designing the so-called \textit{focusing matrix},
which focuses the array steering vectors at each frequency to a reference
frequency, denoted by $f_{0}$ \cite{ICCMwhole,ICCMearly}. From \eqref{eq:Ahhb},
we see that after being processed by the RF network, the subcarrier
frequency $f_{m_{\textrm{f}}}$ in \eqref{eq:horizonarray} is transformed
into the Bessel function, $J_{m_{\textrm{hd}}}\left(\gamma_{m_{\textrm{f}}}(\theta_{k_{m_{\textrm{b}}}})\right)$,
which only depends on $f_{m_{\textrm{f}}}$ and $\theta_{k_{m_{\textrm{b}}}}$,
and is independent of $\phi_{k_{m_{\textrm{b}}}}$. We only need to
optimize the focusing matrices in the elevation angular domain, since
$f_{m_{\textrm{f}}}$ is decoupled from the azimuth angle $\phi_{k_{m_{\textrm{b}}}}$
in \eqref{eq:Ahhb}. Moreover, by taking the vertical array steering
vector in \eqref{eq:Avhb} into consideration, we find that both $J_{m_{\textrm{hd}}}\left(\gamma_{m_{\textrm{f}}}(\theta_{k_{m_{\textrm{b}}}})\right)$
and $a_{\textrm{vhb},m_{\textrm{f}},m_{\textrm{b}}}(\theta_{k_{m_{\textrm{b}}}})$
depend only on the elevation angle $\theta_{k_{m_{\textrm{b}}}}$.

We first design the optimization problem for the horizontal array
steering vectors in \eqref{eq:Ahhb}. Because the measurement samples
in \eqref{eq:tensormodel} are collected from the $M_{\textrm{b}}$
vertical sweeping beams, the optimization can be conducted in each
vertical angular sweeping interval separately. Define
\begin{equation}
\boldsymbol{g}_{m_{\textrm{f}}}(\theta)=\left[J_{-P}\left(\gamma_{m_{\textrm{f}}}(\theta)\right),J_{-P+1}\left(\gamma_{m_{\textrm{f}}}(\theta)\right),\ldots,J_{P}\left(\gamma_{m_{\textrm{f}}}(\theta)\right)\right]^{T},
\end{equation}
which collects all the Bessel functions in \eqref{eq:Ahhb} at the
$m_{\textrm{f}}$-th subcarrier. We discretize each sweeping interval
into $N_{\textrm{b}}$ elevation angular values. Then, the horizontal
factor matrices associated with the subcarrier frequency, $f_{m_{\textrm{f}}}$,
for the $m_{\textrm{b}}$-th sweeping interval can be written as:
\begin{equation}
\mathbf{G}_{\textrm{h},m_{\textrm{f}},m_{\textrm{b}}}=\left[\boldsymbol{g}_{m_{\textrm{f}}}(\theta_{m_{\textrm{b}},1}),\ldots,\boldsymbol{g}_{m_{\textrm{f}}}(\theta_{m_{\textrm{b}},N_{\textrm{b}}})\right],
\end{equation}
where $\theta_{m_{\textrm{b}},n_{\textrm{b}}}=\frac{\pi}{M_{\textrm{b}}}(m_{\textrm{b}}-1)+\frac{\pi}{M_{\textrm{b}}N_{\textrm{b}}}(n_{\textrm{b}}-1)$
is the discretized elevation angle.

We directly use $a_{\textrm{vhb},m_{\textrm{f}},m_{\textrm{b}}}(\theta_{k_{m_{\textrm{b}}}})$
to optimize the vertical array steering vectors by constructing
\begin{equation}
\mathbf{G}_{\textrm{v},m_{\textrm{f}},m_{\textrm{b}}}=\left[\mathbf{a}_{\textrm{vhb},m_{\textrm{f}}}(\theta_{m_{\textrm{b}},1}),\ldots,\mathbf{a}_{\textrm{vhb},m_{\textrm{f}}}(\theta_{m_{\textrm{b}},N_{\textrm{b}}})\right].
\end{equation}

We then obtain the focusing matrices on the vertical and horizontal
planes, denoted by $\mathbf{T}_{\textrm{v},m_{\textrm{f}},m_{\textrm{b}}}$
and $\mathbf{T}_{\textrm{h},m_{\textrm{f}},m_{\textrm{b}}}$, by formulating
the following optimization problems:
\begin{align}
\mathbf{T}_{\textrm{v},m_{\textrm{f}},m_{\textrm{b}}} & =\arg\min_{\mathbf{T}_{\textrm{v},m_{\textrm{f}},m_{\textrm{b}}}}\left\Vert \mathbf{T}_{\textrm{v},m_{\textrm{f}},m_{\textrm{b}}}\mathbf{G}_{\textrm{v},m_{\textrm{f}},m_{\textrm{b}}}-\mathbf{G}_{\textrm{v},m_{\textrm{f}0},m_{\textrm{b}}}\right\Vert _{\textrm{F}}^{2},\nonumber \\
 & \quad\textrm{s.t.}\quad\mathbf{T}_{\textrm{v},m_{\textrm{f}},m_{\textrm{b}}}^{H}\mathbf{T}_{\textrm{v},m_{\textrm{f}},m_{\textrm{b}}}=\mathbf{I}_{M_{\textrm{v}}};\label{eq:wideband1}
\end{align}
\begin{align}
\mathbf{T}_{\textrm{h},m_{\textrm{f}},m_{\textrm{b}}} & =\arg\min_{\mathbf{T}_{\textrm{h},m_{\textrm{f}},m_{\textrm{b}}}}\left\Vert \mathbf{T}_{\textrm{h},m_{\textrm{f}},m_{\textrm{b}}}\mathbf{G}_{\textrm{h},m_{\textrm{f}},m_{\textrm{b}}}-\mathbf{G}_{\textrm{h},m_{\textrm{f}0},m_{\textrm{b}}}\right\Vert _{\textrm{F}}^{2},\nonumber \\
 & \quad\textrm{s.t.}\quad\mathbf{T}_{\textrm{h},m_{\textrm{f}},m_{\textrm{b}}}^{H}\mathbf{T}_{\textrm{h},m_{\textrm{f}},m_{\textrm{b}}}=\mathbf{I}_{2P+1},\label{eq:wideband2}
\end{align}
where $m_{\textrm{f}0}$ is the index to the subcarriers at the reference
frequency $f_{0}$, and the constraints prevent focusing losses \cite{UCAM}.

The solutions to Problems \eqref{eq:wideband1} and \eqref{eq:wideband2}
are given by \cite{ICCMwhole}
\[
\mathbf{T}_{\textrm{v},m_{\textrm{f}},m_{\textrm{b}}}=\mathbf{V}_{\textrm{v},m_{\textrm{f}},m_{\textrm{b}}}\mathbf{U}_{\textrm{v},m_{\textrm{f}},m_{\textrm{b}}}^{H};
\]
\begin{equation}
\mathbf{T}_{\textrm{h},m_{\textrm{f}},m_{\textrm{b}}}=\mathbf{V}_{\textrm{h},m_{\textrm{f}},m_{\textrm{b}}}\mathbf{U}_{\textrm{h},m_{\textrm{f}},m_{\textrm{b}}}^{H},\label{eq:ti}
\end{equation}
where the columns of $\mathbf{U}_{\textrm{v},m_{\textrm{f}},m_{\textrm{b}}}$
$($or $\mathbf{U}_{\textrm{h},m_{\textrm{f}},m_{\textrm{b}}})$ and
$\mathbf{V}_{\textrm{v},m_{\textrm{f}},m_{\textrm{b}}}$ $($or $\mathbf{V}_{\textrm{h},m_{\textrm{f}},m_{\textrm{b}}})$
are the left and right singular vectors of $\mathbf{G}_{\textrm{v},m_{\textrm{f}},m_{\textrm{b}}}\mathbf{G}_{\textrm{v},m_{\textrm{f0}},m_{\textrm{b}}}^{H}$
$($or $\mathbf{G}_{\textrm{h},m_{\textrm{f}},m_{\textrm{b}}}\mathbf{G}_{\textrm{h},m_{\textrm{f0}},m_{\textrm{b}}}^{H})$,
respectively.

We construct $\widetilde{b}_{\textrm{f},m_{\textrm{b}}}=b_{\textrm{f},m_{\textrm{b}}}^{-1}$,
$\widetilde{\mathbf{B}}_{\textrm{v},m_{\textrm{f}},m_{\textrm{b}}}=\mathbf{B}_{\textrm{vdb},m_{\textrm{f}},m_{\textrm{b}}}^{-1},$
and $\mathbf{\widetilde{B}}_{\textrm{h},m_{\textrm{f}},m_{\textrm{b}}}=\mathbf{B}_{\textrm{hdb},m_{\textrm{f}},m_{\textrm{b}}}^{-1}$
to offset the impact of beam sweeping on the received signals. The
focusing matrices \eqref{eq:ti} suppress the frequency dependence
of the array steering vectors. After this coherent wideband processing,
in the $m_{\textrm{b}}$-th sweeping beam, the received signal at
the $m_{\textrm{f}}$-th subcarrier in \eqref{eq:tensormodel} can
be calculated as
\begin{align}
\mathcal{\widetilde{\mathcal{X}}}_{:,:,m_{\textrm{f}},:,m_{\textrm{b}}} & =\mathcal{\mathcal{X}}_{:,:,m_{\textrm{f}},:,m_{\textrm{b}}}\widetilde{b}_{\textrm{f},m_{\textrm{b}}}\times_{1}\left(\mathbf{T}_{\textrm{v},m_{\textrm{f}},m_{\textrm{b}}}\widetilde{\mathbf{B}}_{\textrm{v},m_{\textrm{f}},m_{\textrm{b}}}\right)\nonumber \\
 & \times_{2}\left(\mathbf{T}_{\textrm{h},m_{\textrm{f}},m_{\textrm{b}}}\mathbf{\widetilde{B}}_{\textrm{h},m_{\textrm{f}},m_{\textrm{b}}}\right).\label{COHERENT}
\end{align}

The elements of $\mathcal{\widetilde{\mathcal{X}}}_{:,:,m_{\textrm{f}},:,m_{\textrm{b}}}$
can be expressed as
\begin{align}
 & \widetilde{x}_{m_{\textrm{vd}},m_{\textrm{hd}},m_{\textrm{f}},m_{\textrm{t}},m_{\textrm{b}}}=\sum_{^{k_{m_{\textrm{b}}}=1}}^{K_{m_{\textrm{b}}}}s_{m_{\textrm{t}},k_{m_{\textrm{b}}}}\widetilde{a}_{\textrm{hhb},m_{\textrm{hd}},m_{\textrm{b}}}(\theta_{k_{m_{\textrm{b}}}},\phi_{k_{m_{\textrm{b}}}})\nonumber \\
 & \times\widetilde{a}_{\textrm{vhb},m_{\textrm{vd}},m_{\textrm{b}}}(\theta_{k_{m_{\textrm{b}}}})a_{\textrm{f},m_{\textrm{f}}}(\tau_{k_{m_{\textrm{b}}}})+\widetilde{n}_{m_{\textrm{vd}},m_{\textrm{hd}},m_{\textrm{f}},m_{\textrm{t}},m_{\textrm{b}}},\label{eq:ELEMENTCOHERENT}
\end{align}
where $\widetilde{a}_{\textrm{vhb},m_{\textrm{vd}},m_{\textrm{b}}}(\theta_{k_{m_{\textrm{b}}}})$
and $\widetilde{a}_{\textrm{hhb},m_{\textrm{hd}},m_{\textrm{b}}}(\theta_{k_{m_{\textrm{b}}}},\phi_{k_{m_{\textrm{b}}}})$
are the resultant array manifolds in \eqref{eq:measurementsample}.
$\widetilde{n}_{m_{\textrm{vd}},m_{\textrm{hd}},m_{\textrm{f}},m_{\textrm{t}},m_{\textrm{b}}}$
is the transformed noise sample, which still yields the zero-mean
Gaussian distribution due to the constraints on the beamforming weights
and focusing matrices.

We note that there are two-dimensional variables, $\phi_{k_{m_{\textrm{b}}}}$
and $\theta_{k_{m_{\textrm{b}}}}$, in the frequency-dependent array
steering vectors $\mathbf{a}_{\textrm{v},m_{\textrm{f}},m_{\textrm{b}}}(\theta_{k_{m_{\textrm{b}}}})$
and $\mathbf{a}_{\textrm{h},m_{\textrm{f}},m_{\textrm{b}}}(\theta_{k_{m_{\textrm{b}}}},\phi_{k_{m_{\textrm{b}}}})$.
UCAMI \cite{UCAM} would have to optimize the focusing matrices on
the elevation and azimuth angular domains simultaneously, resulting
in a two-dimensional problem with a high complexity. In contrast,
our proposed method only needs a one-dimensional optimization problem,
i.e., \eqref{eq:wideband1} and \eqref{eq:wideband2}, reducing the
complexity significantly.

\section{Tensor-based Parameter Estimation}

With the received signals preprocessed (in Sections III and IV), the
resultant array steering vectors are frequency-independent in \eqref{eq:ELEMENTCOHERENT}.
Only the delay-dependent factor, $a_{\textrm{f},m_{\textrm{f}}}(\tau_{k_{m_{\textrm{b}}}})$,
depends on the carrier frequency. In this section, we propose a new
tensor-based joint delay-angle estimation algorithm which is the last
step in Fig. 1, and a new spatial smoothing method which is the second-to-last
(optional) step in the figure. Despite the use of the existing truncated
HOSVD, the proposed joint delay-angle estimation algorithm involves
new estimation processes. Specifically, the matrix TLS problem formulation
is generalized to the tensor case. The azimuth angles are estimated
by substituting the estimated elevation angles, which avoids potential
mismatches between the estimated results of the elevation and azimuth
AOAs. By revealing and exploiting the recurrence relations between
the UCAs at different layers of the UCyA, the proposed spatial smoothing
method is developed to decorrelate the coherent signals to correctly
decompose the signal and noise subspaces in all dimensions. The computational
complexity of the proposed algorithm is analyzed at the end.

\subsection{Truncated HOSVD Model of Measurement Samples}

With no a-priori knowledge of the number of signals in each sweeping
beam, $K_{m_{\textrm{b}}}$, we collect all the sweeping results in
\eqref{eq:ELEMENTCOHERENT} to jointly process the signals from the
$K$ signal sources. The element of the received signal tensor model
is given by
\begin{align}
 & y_{m_{\textrm{v}},m_{\textrm{p}},m_{\textrm{f}},m_{\textrm{t}}}=\sum_{m_{\textrm{b}}=1}^{M_{\textrm{b}}}\widetilde{x}_{m_{\textrm{vd}},m_{\textrm{hd}},m_{\textrm{f}},m_{\textrm{t}},m_{\textrm{b}}}\nonumber \\
 & =\sum_{m_{\textrm{b}}=1}^{M_{\textrm{b}}}\left[\sum_{^{k_{m_{\textrm{b}}}=1}}^{K_{m_{\textrm{b}}}}\widetilde{a}_{\textrm{vhb},m_{\textrm{vd}},m_{\textrm{b}}}(\theta_{k_{m_{\textrm{b}}}})\widetilde{a}_{\textrm{hhb},m_{\textrm{hd}},m_{\textrm{b}}}(\theta_{k_{m_{\textrm{b}}}},\phi_{k_{m_{\textrm{b}}}})\right.\nonumber \\
 & \left.\times a_{\textrm{f},m_{\textrm{f}}}(\tau_{k_{m_{\textrm{b}}}})s_{m_{\textrm{t}},k_{m_{\textrm{b}}}}+\widetilde{n}_{m_{\textrm{v}},m_{\textrm{p}},m_{\textrm{f}},m_{\textrm{b}},m_{\textrm{t}}}\vphantom{\sum_{^{k_{m_{\textrm{b}}}=1}}^{K_{m_{\textrm{b}}}}}\right]=\sum_{^{k=1}}^{K}\widetilde{a}_{\textrm{vhb},m_{\textrm{vd}}}(\theta_{k})\nonumber \\
 & \times\widetilde{a}_{\textrm{hhb},m_{\textrm{hd}}}(\theta_{k},\phi_{k})a_{\textrm{f},m_{\textrm{f}}}(\tau_{k})s_{m_{\textrm{t}},k}+\dot{n}_{m_{\textrm{v}},m_{\textrm{p}},m_{\textrm{f}},m_{\textrm{t}}},\label{eq:element}
\end{align}
which can be expressed concisely as:
\begin{equation}
\mathcal{Y}=\sum_{m_{\textrm{b}}=1}^{M_{\textrm{b}}}\mathcal{\widetilde{\mathcal{X}}}_{:,:,:,:,m_{\textrm{b}}}=\widetilde{\mathcal{A}}\times_{4}\mathbf{S}+\mathcal{\dot{N}}\in\mathbb{C}^{M_{\textrm{vd}}\times M_{\textrm{hd}}\times M{}_{\textrm{f}}\times M_{\textrm{t}}},\label{eq:modeltensorfinal}
\end{equation}
where $\mathbf{S}=[\mathbf{S}_{1},\mathbf{S}_{2},\ldots,\mathbf{S}_{M_{\textrm{b}}}]\in\mathbb{C}^{M_{\textrm{t}}\times K}$
collects all the symbols and $\mathcal{\dot{N}}=\sum_{m_{\textrm{b}}=1}^{M_{\textrm{b}}}\mathcal{N}_{m_{\textrm{b}}}$
collects all noise samples. An illustration of the received signal
tensor model is shown in Fig. \ref{fig:(a)-System-configuration;}(b).
In \eqref{eq:modeltensorfinal}, $\widetilde{\mathcal{A}}\in\mathbb{C}^{M_{\textrm{vd}}\times M_{\textrm{hd}}\times M{}_{\textrm{f}}\times K}$
is known as the space-time response tensor \cite{SCIVanderveen},
and obtained by concatenating the $K$ response tensors, $\mathcal{\widetilde{\mathcal{A}}}_{\textrm{\ensuremath{k}}}\in\mathbb{C}^{M_{\textrm{vd}}\times M_{\textrm{hd}}\times M{}_{\textrm{f}}}$,
as given by:
\begin{equation}
\widetilde{\mathcal{A}}=\left[\widetilde{\mathcal{A}}_{1}\sqcup_{4}\widetilde{\mathcal{A}}_{2}\sqcup_{4}\ldots\sqcup_{4}\mathcal{\widetilde{\mathcal{A}}}_{\textrm{\ensuremath{K}}}\right].\label{eq:STresponsetensor}
\end{equation}
Because the array steering vectors are frequency-independent after
the coherent wideband preprocessing (as described in Section IV),
the space-time response tensor of the $k$-th signal source, $\mathcal{\widetilde{\mathcal{A}}}_{\textrm{\ensuremath{k}}}$,
is given by
\begin{equation}
\mathcal{\widetilde{\mathcal{A}}}_{\textrm{\ensuremath{k}}}=\mathbf{\widetilde{a}}_{\textrm{vhb}}(\theta_{k})\circ\widetilde{\mathbf{a}}_{\textrm{hhb}}(\theta_{k},\phi_{k})\circ\mathbf{a}_{\textrm{f}}(\tau_{k}),\label{eq:signaltensor}
\end{equation}
where $\left[\mathbf{\widetilde{a}}_{\textrm{vhb}}(\theta_{k})\right]_{m_{\textrm{vd}},1}=\widetilde{a}_{\textrm{vdb},m_{\textrm{vd}}}(\theta_{k})$,
$\left[\mathbf{\widetilde{a}}_{\textrm{hhb}}(\theta_{k},\phi_{k})\right]_{m_{\textrm{hd}},1}=\widetilde{a}_{\textrm{hhb},m_{\textrm{hd}}}(\theta_{k},\phi_{k}),$
and $\left[\mathbf{a}_{\textrm{f}}(\tau_{k})\right]_{m_{\textrm{f}},1}=a_{\textrm{f},m_{\textrm{f}}}(\tau_{k})$.

By substituting \eqref{eq:STresponsetensor} and \eqref{eq:signaltensor}
into \eqref{eq:modeltensorfinal}, we obtain
\begin{equation}
\mathcal{Y}=\sum_{^{k=1}}^{K}\mathbf{\widetilde{a}}_{\textrm{vhb}}(\theta_{k})\circ\widetilde{\mathbf{a}}_{\textrm{hhb}}(\theta_{k},\phi_{k})\circ\mathbf{a}_{\textrm{f}}(\tau_{k})\circ\mathbf{s}_{k}+\mathcal{\dot{N}},\label{eq:tensorrank}
\end{equation}
where $\mathbf{s}_{k}=[\mathbf{S}]_{:,k}$. \eqref{eq:tensorrank}
indicates that, in a noiseless case, $\mathcal{Y}$ can be regarded
as the sum of $K$ rank-one tensors. Therefore, \eqref{eq:tensorrank}
is the CP decomposition of $\mathcal{Y}$ (see Property 3 in Appendix
I). $\textrm{Rank}(\mathcal{Y})=K$\footnote{According to \eqref{eq:tensorrank}, we have $\textrm{Rank}(\mathcal{Y})\leq K$.
$\textrm{Rank}(\mathcal{Y})<K$ only occurs when the locations of
two coherent sources are the same, which rarely happens. }. According to Property 3 in Appendix I, \eqref{eq:tensorrank} can
be written as
\begin{align}
\mathcal{Y} & =\left\llbracket \mathcal{Z}_{\textrm{s}};\mathbf{\widetilde{A}}_{\textrm{vhb}},\mathbf{\widetilde{A}}_{\textrm{hhb}},\mathbf{A}_{\textrm{f}},\mathbf{S}\right\rrbracket +\mathcal{\dot{N}}\label{eq:cp}
\end{align}
where $\left[\mathbf{\widetilde{A}}_{\textrm{vhb}}\right]_{:,k}=\mathbf{\widetilde{a}}_{\textrm{vhb}}(\theta_{k}),$
$\left[\mathbf{\widetilde{A}}_{\textrm{hhb}}\right]_{:,k}=\widetilde{\mathbf{a}}_{\textrm{hhb}}(\theta_{k},\phi_{k})$,
$\left[\mathbf{A}_{\textrm{f}}\right]_{:,k}=\mathbf{a}_{\textrm{f}}(\tau_{k})$,
and $\mathcal{Z}_{\textrm{s}}\in\mathbb{C}^{K\times K\times K\times K}$
is an identity superdiagonal tensor.

Given the typically sparse multipath propagation of mmWave, the number
of received paths is much smaller than the numbers of antennas, subcarriers,
and time frames, i.e., $K<\min(M_{\textrm{vd}},M_{\textrm{hd}},$
$M_{\textrm{f}},M_{\textrm{t}})$. Thus, the ranks of $\mathbf{\widetilde{A}}_{\textrm{vhb}}$,
$\mathbf{\widetilde{A}}_{\textrm{hhb}}$, $\mathbf{A}_{\textrm{f}}$
and $\mathbf{S}$ are all $K$. According to the CP model \eqref{eq:cp},
in the presence of non-negligible noises, $\mathbf{\widetilde{A}}_{\textrm{vhb}}$,
$\mathbf{\widetilde{A}}_{\textrm{hhb}}$, $\mathbf{A}_{\textrm{f}}$
and $\mathbf{S}$ correspond to the factor matrix of the measurement
tensor $\mathcal{Y}$ in each mode. The ranks of the mode-$n$ unfoldings
of tensor $\mathcal{Y}$, i.e., the $n$-ranks of $\mathcal{Y}$ $(n=1,2,3,4)$,
are all $K$.

As a high-dimensional generalization of matrix SVD, the HOSVD (see
Property 2 in Appendix I) conducts the SVD of the unfolding of $\mathcal{Y}$
in each mode separately. This can suppress the received noise in each
mode. The HOSVD of the measurement tensor $\mathcal{Y}$ is given
by
\begin{equation}
\mathcal{Y}=\mathcal{L}\times_{1}\mathbf{U}_{\textrm{v}}\times_{2}\mathbf{U}_{\textrm{h}}\times_{3}\mathbf{U}_{\textrm{f}}\times_{4}\mathbf{U}_{\textrm{t}}=\left\llbracket \mathcal{L};\mathbf{U}_{\textrm{v}},\mathbf{U}_{\textrm{h}},\mathbf{U}_{\textrm{f}},\mathbf{U}_{\textrm{t}}\right\rrbracket ,\label{eq:hosvd}
\end{equation}
where the unitary matrices, $\mathbf{U}_{\textrm{v}}\in\mathbb{C}^{M_{\textrm{vd}}\times M_{\textrm{vd}}}$,
$\mathbf{U}_{\textrm{h}}\in\mathbb{C}^{M_{\textrm{hd}}\times M_{\textrm{hd}}}$,
$\mathbf{U}_{\textrm{f}}\in\mathbb{C}^{M_{\textrm{f}}\times M_{\textrm{f}}}$,
and $\mathbf{U}_{\textrm{t}}\in\mathbb{C}^{M_{\textrm{t}}\times M_{\textrm{t}}}$,
are the left singular matrices of the mode-$n$ unfoldings of tensor
$\mathcal{Y},$ and the core tensor $\mathcal{L}\in\mathbb{C}^{M_{\textrm{vd}}\times M_{\textrm{hd}}\times M{}_{\textrm{f}}\times M_{\textrm{t}}}$
is obtained by moving the singular matrices to the left-hand side
of \eqref{eq:hosvd}:
\begin{equation}
\mathcal{L}=\mathcal{Y}\times_{1}\mathbf{U}_{\textrm{v}}^{H}\times_{2}\mathbf{U}_{\textrm{h}}^{H}\times_{3}\mathbf{U}_{\textrm{f}}^{H}\times_{4}\mathbf{U}_{\textrm{t}}^{H}.
\end{equation}

Because the $n$-ranks of tensor $\mathcal{Y}$ are $K$, the SVD
of the mode-1 unfolding $\mathbf{Y}_{(1)}\in\mathbb{C}^{M_{\textrm{vd}}\times(M/M_{\textrm{\textrm{v}d}})}$
can be written as
\begin{align}
 & \mathbf{Y}_{(1)}=\mathbf{U}_{\textrm{v}}\mathbf{\mathbf{\Sigma}}_{\textrm{v}}\mathbf{\mathbf{V}}_{\textrm{v}}^{H}=\nonumber \\
 & \left[\mathbf{\mathbf{U}}_{\textrm{v},\textrm{s}}\;\mathbf{U}_{\textrm{v},\textrm{n}}\right]\left[\begin{array}{cc}
\mathbf{\mathbf{\Sigma}}_{\textrm{v},\textrm{s}} & \mathbf{0}_{K\times(\frac{M}{M_{\textrm{\textrm{v}d}}}-K_{\textrm{vd}})}\\
\mathbf{0}_{(M_{\textrm{vd}}-K_{\textrm{vd}})\times K_{\textrm{vd}}} & \mathbf{\mathbf{\Sigma}}_{\textrm{v},\textrm{n}}
\end{array}\right]\left[\mathbf{\mathbf{V}}_{\textrm{v},\textrm{s}}\;\mathbf{V}_{\textrm{v},\textrm{n}}\right]^{H},\label{eq:subspace}
\end{align}
where $K_{\textrm{vd}}=\min(K,M_{\textrm{vd}})$ and $M=M_{\textrm{vd}}M_{\textrm{hd}}M{}_{\textrm{f}}M_{\textrm{t}}$.
The signal subspace $\mathbf{\mathbf{U}}_{\textrm{v},\textrm{s}}\in\mathbb{C}^{M_{\textrm{vd}}\times K_{\textrm{vd}}}$
and the noise subspace $\mathbf{U}_{\textrm{v},\textrm{n}}\in\mathbb{C}^{M_{\textrm{vd}}\times(M_{\textrm{vd}}-K_{\textrm{vd}})}$
of the mode-1 unfolding $\mathbf{Y}_{(1)}$ correspond to the $K_{\textrm{vd}}$
largest and the $(M_{\textrm{vd}}-K_{\textrm{vd}})$ smallest elements
of the diagonal matrix $\mathbf{\mathbf{\Sigma}}_{\textrm{v}}=\textrm{diag}(\sigma_{\textrm{v},1},\sigma_{\textrm{v},2},\ldots,$
$\sigma_{\textrm{v},M_{\textrm{vd}}})$, respectively. $\sigma_{\textrm{v},1},\sigma_{\textrm{v},2},\ldots,\sigma_{\textrm{v},M_{\textrm{vd}}}$
are the non-zero singular values of the mode-1 unfolding $\mathbf{Y}_{(1)}$,
and calculated by $\sigma_{\textrm{v},m_{\textrm{vd}}}=\left\Vert \mathcal{L}_{m_{\textrm{vd}},:,:,:,:}\right\Vert $.
The signal subspace matrices of the mode-2,3,4 unfoldings of $\mathcal{Y}$,
i.e., $\mathbf{U}_{\textrm{h,s}}\in\mathbb{C}^{M_{\textrm{hd}}\times K_{\textrm{hd}}}$
, $\mathbf{U}_{\textrm{f,s}}\in\mathbb{C}^{M_{\textrm{f}}\times K_{\textrm{f}}}$,
and $\mathbf{U}_{\textrm{t,s}}\in\mathbb{C}^{M_{\textrm{t}}\times K}$
can be obtained in the same way, where $K_{\textrm{hd}}=\min(K,M_{\textrm{hd}})$
and $K_{\textrm{f}}=\min(K,M_{\textrm{f}})$ .

By removing the noise subspace in each mode of $\mathcal{Y}$, we
construct a low-rank truncated HOSVD model of the noise-free measurement
tensor $\mathcal{Y}_{\textrm{s}}$ \cite{TensorDecompositions1},
as given by
\begin{equation}
\mathcal{Y}_{\textrm{s}}=\mathcal{L}_{\textrm{s}}\times_{1}\mathbf{U}_{\textrm{v,s}}\times_{2}\mathbf{U}_{\textrm{h,s}}\times_{3}\mathbf{U}_{\textrm{f,s}}\times_{4}\mathbf{U}_{\textrm{t,s}}\in\mathbb{C}^{M_{\textrm{vd}}\times M_{\textrm{hd}}\times M{}_{\textrm{f}}\times M_{\textrm{t}}},\label{eq:truncatedHOSVD}
\end{equation}
where $\mathcal{L}_{\textrm{s}}\in\mathbb{C}^{K_{\textrm{vd}}\times K_{\textrm{hd}}\times K_{\textrm{f}}\times K}$
is obtained by discarding the insignificant singular values of the
mode-$n$ unfoldings of $\mathcal{Y}$.

\subsection{Joint Angle-Delay Estimation }

We propose a tensor-based joint delay-angle estimation algorithm by
exploiting the shift-invariance relations between the elements in
each mode of the measurement tensor. By comparing \eqref{eq:modeltensorfinal}
with \eqref{eq:cp}, we first obtain
\begin{equation}
\widetilde{\mathcal{A}}=\mathcal{Z}_{\textrm{s}}\times_{1}\mathbf{\widetilde{A}}_{\textrm{vhb}}\times_{2}\mathbf{\widetilde{A}}_{\textrm{hhb}}\times_{3}\mathbf{A}_{\textrm{f}}.\label{eq:A}
\end{equation}
According to the truncated HOSVD model \eqref{eq:truncatedHOSVD},
we define the signal subspace tensor:
\begin{equation}
\mathcal{U}_{\textrm{s}}=\mathcal{L}_{\textrm{s}}\times_{1}\mathbf{U}_{\textrm{v,s}}\times_{2}\mathbf{U}_{\textrm{h,s}}\times_{3}\mathbf{U}_{\textrm{f,s}}\in\mathbb{C}^{M_{\textrm{vd}}\times M_{\textrm{hd}}\times M{}_{\textrm{f}}\times K}.\label{eq:Us}
\end{equation}
By comparing \eqref{eq:modeltensorfinal}, \eqref{eq:truncatedHOSVD},
\eqref{eq:A} and \eqref{eq:Us}, we have $\mathcal{U}_{\textrm{s}}\times_{4}\mathbf{U}_{\textrm{t,s}}=\widetilde{\mathcal{A}}\times_{4}\mathbf{S}.$
Because $\mathbf{U}_{\textrm{t,s}}\in\mathbb{C}^{M_{\textrm{t}}\times K}$
and $\mathbf{S}\in\mathbb{C}^{M_{\textrm{t}}\times K}$ are full column
rank matrices, we obtain
\begin{equation}
\widetilde{\mathcal{A}}=\mathcal{U}_{\textrm{s}}\times_{4}\mathbf{D},\label{eq:UsA}
\end{equation}
where $\mathbf{D}\in\mathbb{C}^{K\times K}$ is a full rank matrix.
Based on \eqref{eq:UsA}, we generalize the matrix-based subspace
algorithm to the tensor, and estimate the delay and angles of each
signal path.

\subsubsection{Estimation of Elevation Angle}

We first propose a tensor-based total-least-squares ESPRIT (TLS-ESPRIT)
algorithm to estimate the elevation angle and delay. To estimate the
elevation angle of each signal path, we first reveal and then exploit
the shift-invariance relations underlying the vertical array steering
matrix $\mathbf{\widetilde{A}}_{\textrm{v}}$, according to \eqref{eq:verticalarray}
and \eqref{eq:element}.

To select the elevation angle-related subtensors, we define two selection
matrices:
\[
\mathbf{J}_{\textrm{v}1}=[\mathbf{I}_{M_{\textrm{vd}}-1},\mathbf{0}_{(M_{\textrm{vd}}-1)\times1}]\in\mathbb{R}^{(M_{\textrm{vd}}-1)\times M_{\textrm{vd}}};
\]
\begin{equation}
\mathbf{J}_{\textrm{v}2}=[\mathbf{0}_{(M_{\textrm{vd}}-1)\times1},\mathbf{I}_{M_{\textrm{vd}}-1}]\in\mathbb{R}^{(M_{\textrm{vd}}-1)\times M_{\textrm{vd}}},\label{eq:Jv1}
\end{equation}
which are two auxiliary matrices. We reveal the following shift-invariance
relation among the selected subtensors:
\begin{equation}
\widetilde{\mathcal{A}}\times_{1}\mathbf{J}_{\textrm{v}2}=\widetilde{\mathcal{A}}\times_{1}\mathbf{J}_{\textrm{v}1}\times_{4}\mathbf{\mathbf{\Theta}_{\textrm{v}}},\label{eq:AJv}
\end{equation}
where $\mathbf{\mathbf{\Theta}_{\textrm{v}}}=\textrm{diag}\left(e^{-j\frac{2\pi}{c}f_{0}h\cos(\theta_{1})},\ldots,e^{-j\frac{2\pi}{c}f_{0}h\cos(\theta_{K})}\right)\in\mathbb{C}^{K\times K}$.
The shift-invariance relation is the key to our design of the following
tensor-based TLS-ESPRIT algorithm\footnote{The least-squares (LS) procedure can also be used for solving the
invariance equation \eqref{eq:Usv12}, but has slightly lower accuracy
than TLS. Section VI will provide the results of performance comparison
between the proposed algorithm (T-CTLS), which applies TLS-ESPRIT
for parameter estimation, with its variation (T-CLS), which uses LS-ESPRIT.}. The algorithm estimates the elevation angle of each signal in the
tensor form.

By substituting \eqref{eq:UsA} into \eqref{eq:AJv}, we have
\begin{equation}
\mathcal{U}_{\textrm{sv}2}\times_{4}\mathbf{D}=\mathcal{U}_{\textrm{sv}1}\times_{4}\left(\mathbf{\mathbf{\Theta}_{\textrm{v}}}\mathbf{D}\right),\label{eq:UsJv}
\end{equation}
where $\mathcal{U}_{\textrm{sv}1}=\mathcal{U}_{\textrm{s}}\times_{1}\mathbf{J}_{\textrm{v}1}\in\mathbb{C}^{(M_{\textrm{vd}}-1)\times M_{\textrm{hd}}\times M{}_{\textrm{f}}\times K}$
and $\mathcal{U}_{\textrm{sv}2}=\mathcal{U}_{\textrm{s}}\times_{1}\mathbf{J}_{\textrm{v}2}\in\mathbb{C}^{(M_{\textrm{vd}}-1)\times M_{\textrm{hd}}\times M{}_{\textrm{f}}\times K}$
are the selected subtensors of the signal subspace tensor $\mathcal{U}_{\textrm{s}}$.
Since $\mathbf{D}$ is a full rank matrix, we can left-multiply its
inverse to both sides of \eqref{eq:UsJv} and obtain
\begin{equation}
\mathcal{U}_{\textrm{sv}2}=\mathcal{U}_{\textrm{sv}1}\times_{4}\mathbf{\Psi}_{\textrm{v}},\label{eq:Usv12}
\end{equation}
where $\mathbf{\Psi}_{\textrm{v}}=\mathbf{D}^{-1}\mathbf{\mathbf{\Theta}_{\textrm{v}}}\mathbf{D}\in\mathbb{C}^{K\times K}$.

To obtain the estimate of $\mathbf{\Psi}_{\textrm{v}}$ in \eqref{eq:Usv12},
we define $\mathbf{\Upsilon}_{\textrm{v}}=\left[\mathbf{\Upsilon}_{\textrm{v}1}\quad\mathbf{\Upsilon}_{\textrm{v2}}\right]\in\mathbb{C}^{K\times2K}$.
According to the standard TLS \cite{Arraysignalprocessingconceptsandtechniques},
the estimate of $\mathbf{\Psi}_{\textrm{v}}$ is $\hat{\mathbf{\Psi}}_{\textrm{v}}=-\hat{\mathbf{\Upsilon}}_{\textrm{v}1}\hat{\mathbf{\Upsilon}}_{\textrm{v}2}^{-1},$
where the $K$ eigenvalues of $\hat{\mathbf{\Psi}}_{\textrm{v}}$,
i.e., $\lambda_{\textrm{v,}k}$, $k=1,2,\ldots,K$, are sorted in
descending order. We now generalize the matrix TLS problem formulation
\cite{Arraysignalprocessingconceptsandtechniques} to the tensor case,
as given by:
\begin{align}
\hat{\mathbf{\Upsilon}}_{\textrm{v}} & =\arg\min_{\mathbf{\Upsilon}_{\textrm{v}}}\left\Vert \mathcal{U}_{\textrm{sv}1}\times_{4}\mathbf{\Upsilon}_{\textrm{v}1}+\mathcal{U}_{\textrm{sv}2}\times_{4}\mathbf{\Upsilon}_{\textrm{v2}}\right\Vert ,\nonumber \\
 & \quad\textrm{s.t.}\quad\mathbf{\Upsilon}_{\textrm{v}}\mathbf{\Upsilon}_{\textrm{v}}^{H}=\mathbf{I}_{K},\label{eq:tls}
\end{align}
which finds a unitary matrix $\mathbf{\Upsilon}_{\textrm{v}}$ whose
submatrices are orthogonal to $\mathcal{U}_{\textrm{sv}1}$ and $\mathcal{U}_{\textrm{sv}2}$
in mode-4.

According to \eqref{eq:propertymultiproduct}, the mode-4 unfoldings
of $\mathcal{U}_{\textrm{sv}1}$ is given by
\begin{equation}
\mathbf{U}_{\textrm{sv}1}{}_{(4)}=\mathbf{U}_{\textrm{s}}{}_{(4)}\left(\mathbf{J}_{\textrm{v}1}\otimes\mathbf{I}_{M_{\textrm{hd}}}\otimes\mathbf{I}_{M_{\textrm{f}}}\right)^{T},
\end{equation}
where $\mathbf{U}_{\textrm{s}}{}_{(4)}\in\mathbb{C}^{K\times M_{\textrm{vd}}M_{\textrm{hd}}M{}_{\textrm{f}}}$
is the mode-4 unfolding of $\mathcal{U}_{\textrm{s}}$. The mode-4
unfoldings of $\mathcal{U}_{\textrm{sv}1}$ can be formulated in the
same way. Since $\left\Vert \mathcal{A}\right\Vert =\left\Vert \mathbf{A}_{(n)}\right\Vert _{\textrm{F}}$
$(n=1,2,\ldots,N)$ \cite{TensorDecompositions1}, we rewrite the
tensor TLS problem \eqref{eq:tls} in a matrix format as:
\begin{align}
\hat{\mathbf{\Upsilon}}_{\textrm{v}} & =\arg\min_{\mathbf{\Upsilon}_{\textrm{v}}}\left\Vert \mathbf{\Upsilon}_{\textrm{v}1}\mathbf{U}_{\textrm{s}}{}_{(4)}\left(\mathbf{J}_{\textrm{v}1}\otimes\mathbf{I}_{M_{\textrm{hd}}}\otimes\mathbf{I}_{M_{\textrm{f}}}\right)^{T}\right.\nonumber \\
 & \left.+\mathbf{\Upsilon}_{\textrm{v2}}\mathbf{U}_{\textrm{s}}{}_{(4)}\left(\mathbf{J}_{\textrm{v}2}\otimes\mathbf{I}_{M_{\textrm{hd}}}\otimes\mathbf{I}_{M_{\textrm{f}}}\right)^{T}\right\Vert _{\textrm{F}}\nonumber \\
 & =\arg\min_{\mathbf{\Upsilon}_{\textrm{v}}}\left\Vert \mathbf{W}_{\textrm{v}}\mathbf{\Upsilon}_{\textrm{v}}^{T}\right\Vert _{\textrm{F}},
\end{align}
where
\begin{align}
\mathbf{W}_{\textrm{v}} & =\left[\left(\mathbf{J}_{\textrm{v}1}\otimes\mathbf{I}_{M_{\textrm{hd}}}\otimes\mathbf{I}_{M_{\textrm{f}}}\right)\mathbf{U}_{\textrm{s}}{}_{(4)}^{T}\quad\left(\mathbf{J}_{\textrm{v}2}\otimes\mathbf{I}_{M_{\textrm{hd}}}\otimes\mathbf{I}_{M_{\textrm{f}}}\right)\mathbf{U}_{\textrm{s}}{}_{(4)}^{T}\right]\nonumber \\
 & \in\mathbb{C}^{(M_{\textrm{vd}}-1)M_{\textrm{hd}}M{}_{\textrm{f}}\times2K}.
\end{align}
The SVD of $\mathbf{W}_{\textrm{v}}^{H}\mathbf{W}_{\textrm{v}}$ is
written as $\mathbf{W}_{\textrm{v}}^{H}\mathbf{W}_{\textrm{v}}=\mathbf{\dot{U}}_{\textrm{v}}\mathbf{\dot{\Lambda}}_{\textrm{v}}\mathbf{\dot{V}}_{\textrm{v}},$
where $\mathbf{\dot{U}}_{\textrm{v}}\in\mathbb{C}^{2K\times2K}$ and
$\mathbf{\dot{V}}_{\textrm{v}}\in\mathbb{C}^{2K\times2K}$ are the
left and right singular matrices, respectively; and $\mathbf{\dot{\Lambda}}_{\textrm{v}}\in\mathbb{C}^{2K\times2K}$
contains singular values. We partition $\mathbf{\dot{U}}_{\textrm{v}}$
into four blocks:
\begin{equation}
\mathbf{\dot{U}}_{\textrm{v}}=\left[\begin{array}{cc}
\mathbf{\dot{U}}_{\textrm{v11}} & \mathbf{\dot{U}}_{\textrm{v12}}\\
\mathbf{\dot{U}}_{\textrm{v21}} & \mathbf{\dot{U}}_{\textrm{v22}}
\end{array}\right]\in\mathbb{C}^{2K\times2K}.\label{eq:TLSfinal}
\end{equation}
Let $\hat{\mathbf{\Upsilon}}_{\textrm{v}1}=\mathbf{\dot{U}}_{\textrm{v12}}^{T}\in\mathbb{C}^{K\times K}$
and $\mathbf{\hat{\Upsilon}}_{\textrm{v}2}=\mathbf{\dot{U}}_{\textrm{v22}}^{T}\in\mathbb{C}^{K\times K}$.

According to the array steering expression in \eqref{eq:verticalarray},
the elevation angle of the $k$-th path can be finally estimated as
\begin{equation}
\hat{\theta}_{k}=\arccos\left(\frac{jc\ln(\lambda_{\textrm{v,}k})}{2\pi f_{0}h}\right).\label{theta}
\end{equation}

\subsubsection{Estimation of Delay}

We can estimate the delays by exploiting the shift-invariance relation
between the delay-related subtensors. We express the delay-dependent
shift-invariance relation, as follows.
\begin{equation}
\widetilde{\mathcal{A}}\times_{3}\mathbf{J}_{\textrm{f}2}=\widetilde{\mathcal{A}}\times_{3}\mathbf{J}_{\textrm{f}1}\times_{4}\mathbf{\mathbf{\Theta}_{\textrm{f}}},
\end{equation}
where $\mathbf{\mathbf{\Theta}_{\textrm{f}}}=\textrm{diag}\left(e^{-j2\pi\Delta_{\textrm{F}}\tau_{1}},\ldots,e^{-j2\pi\Delta_{\textrm{F}}\tau_{K}}\right)$
with $\Delta_{\textrm{F}}$ being the subcarrier spacing. $\mathbf{J}_{\textrm{f}1}$
and $\mathbf{J}_{\textrm{f}2}$ are two selection matrices to select
the delay-related subtensors. $\mathbf{J}_{\textrm{f}1}$ and $\mathbf{J}_{\textrm{f}2}$
can be constructed in the same way as in \eqref{eq:Jv1}. By using
TLS-ESPRIT \eqref{eq:tls}, the delay of the $k$-th path, $\tau_{k}$,
can be estimated as
\begin{equation}
\hat{\tau}_{k}=\frac{j\ln(\lambda_{\textrm{f,}k})}{2\pi\Delta_{\textrm{F}}},\label{delay}
\end{equation}
where $\lambda_{\textrm{f,}k}$ is an eigenvalue of the delay-related
matrix $\mathbf{\Psi}_{\textrm{f}}=\mathbf{D}\mathbf{\mathbf{\Theta}_{\textrm{f}}}\mathbf{D}^{-1}$.
In the presence of non-negligible noises, the estimates of the elevation
angle and delay of each source may be paired incorrectly. After obtaining
the estimates of $\hat{\mathbf{\Psi}}_{\textrm{v}}$ and $\hat{\mathbf{\Psi}}_{\textrm{f}}$
with \eqref{eq:tls}, joint SVD methods \cite{jointSVD} can be used
to obtain the joint eigenvalues of $\hat{\mathbf{\Psi}}_{\textrm{v}}$
and $\hat{\mathbf{\Psi}}_{\textrm{f}}$, and then the correctly matched
pairs of estimated parameters can be obtained.

\subsubsection{Estimation of Azimuth Angle}

We design the tensor-MUSIC algorithm \cite{tensormmWave} to estimate
the azimuth angle of each path. From \eqref{eq:Ahhb}, there are nonlinear
Bessel functions in the expression for the horizontal array steering
matrix $\mathbf{\widetilde{A}}_{\textrm{h}}$, and therefore there
is no shift-invariance relation for the azimuth angle estimation,
as opposed to \eqref{eq:AJv}.

According to \eqref{eq:truncatedHOSVD}, we discard the largest $K$
singular values of the mode-$n$ unfoldings of the measurement tensor
$\mathcal{Y}$, i.e., setting the corresponding parts of $\mathcal{L}$
to zero. Then we obtain the noise subspace tensor as\footnote{It is well known that this solution for estimating the noise subspace
is not optimal in the least squares sense. However, it is a good approximation
in most cases \cite{TensorDecompositions2,tensorMUSIC} and it is
easy to implement.}:
\begin{equation}
\mathcal{Y}_{\textrm{n}}=\mathcal{L}_{\textrm{n}}\times_{1}\mathbf{U}_{\textrm{v,n}}\times_{2}\mathbf{U}_{\textrm{h,n}}\times_{3}\mathbf{U}_{\textrm{f,n}}\times_{4}\mathbf{U}_{\textrm{t,n}},\label{eq:noisesubspace}
\end{equation}
where $\mathbf{U}_{\textrm{v,n}}\in\mathbb{C}^{M_{\textrm{vd}}\times(M_{\textrm{vd}}-K)}$
is constructed by the last $(M_{\textrm{vd}}-K)$ columns of $\mathbf{U}_{\textrm{v}}$;
$\mathbf{U}_{\textrm{h,n}}\in\mathbb{C}^{M_{\textrm{hd}}\times(M_{\textrm{hd}}-K)}$
is the last $(M_{\textrm{hd}}-K)$ columns of $\mathbf{U}_{\textrm{h}}$;
$\mathbf{U}_{\textrm{f,n}}\in\mathbb{C}^{M_{\textrm{f}}\times(M_{\textrm{f}}-K)}$
is the last $(M_{\textrm{f}}-K)$ columns of $\mathbf{U}_{\textrm{f}}$;
and $\mathbf{U}_{\textrm{t,n}}\in\mathbb{C}^{M_{\textrm{t}}\times(M_{\textrm{t}}-K)}$
is the last $(M_{\textrm{t}}-K)$ columns of $\mathbf{U}_{\textrm{t}}$.
The core $\mathcal{L}_{\textrm{n}}$ can be evaluated by
\begin{equation}
\mathcal{L}_{\textrm{n}}=\mathcal{Y}_{\textrm{n}}\times_{1}\mathbf{U}_{\textrm{v,n}}^{H}\times_{2}\mathbf{U}_{\textrm{h,n}}^{H}\times_{3}\mathbf{U}_{\textrm{f,n}}^{H}\times_{4}\mathbf{U}_{\textrm{t,n}}^{H}.
\end{equation}
Based on the subspace estimation of $\mathcal{Y}$ \eqref{eq:Us},
we generalize the matrix-based MUSIC, and the tensor MUSIC spectrum
of the azimuth angle is defined as
\begin{equation}
\textrm{SP}_{\textrm{MUSIC}}(\mathbf{\Phi})=\left\Vert \widetilde{\mathcal{A}}\times_{2}\mathbf{U}_{\textrm{h,n}}^{H}\right\Vert ^{-2},\label{eq:music}
\end{equation}
where $\mathbf{\Phi}=\left[\phi_{1},\phi_{2},\ldots,\phi_{K}\right]$.

According to \eqref{eq:propertymultiproduct}, the mode-2 matricization
of $\widetilde{\mathcal{A}}$ in \eqref{eq:A} can be expressed as
\begin{equation}
\mathbf{\widetilde{A}}_{\textrm{(2)}}=\mathbf{\widetilde{A}}_{\textrm{hhb}}Z_{\textrm{s}(2)}\left(\mathbf{A}_{\textrm{f}}\otimes\mathbf{I}_{M_{\textrm{t}}}\otimes\mathbf{\widetilde{A}}_{\textrm{vhb}}\right)^{T}.\label{eq:A2}
\end{equation}

We substitute \eqref{eq:A2} into \eqref{eq:music} and obtain the
mode-2 matricization of \eqref{eq:music}, as given by
\begin{equation}
\textrm{SP}_{\textrm{MUSIC}}(\mathbf{\Phi})=\left\Vert \mathbf{U}_{\textrm{h,n}}^{H}\mathbf{\widetilde{A}}_{\textrm{hhb}}Z_{\textrm{s}(2)}\left(\mathbf{A}_{\textrm{f}}\otimes\mathbf{I}_{M_{\textrm{t}}}\otimes\mathbf{\widetilde{A}}_{\textrm{vhb}}\right)^{T}\right\Vert _{\textrm{F}}^{-2}.\label{eq:MUSIC2}
\end{equation}
By substituting the estimated elevation angle of each path, i.e.,
\eqref{theta}, into \eqref{eq:MUSIC2}, the corresponding azimuth
angle $\phi_{k}$ can be estimated by searching the prominent peaks
of the tensor MUSIC spectrum \eqref{eq:MUSIC2}.

\begin{remark} When applying the tensor-based TLS-ESPRIT and MUSIC
algorithms to estimate the parameters, we first apply the HOSVD evaluates
the SVD of the unfoldings of $\mathcal{Y}$ in all modes, and then
suppress the noise components by discarding the singular vectors and
slices of the core tensor that correspond to insignificant singular
values of the matricized tensor in each mode. The uniqueness and identifiability
of the proposed algorithm inherits from that of the matrix-based counterpart
of the algorithm, due to the fact that the proposed algorithm can
be regarded as the high-dimensional generalization of the matrix-based
counterpart \cite{TensorDecompositions1}. In particular, to achieve
the unique parameter estimates of the $K$ sources would need to construct
the signal subspace tensor $\mathcal{U}_{\textrm{s}}$ with a smaller
number of sources $K$ than time frames $M_{\textrm{t}}$. Our method
is suitable for multi-dimensional parameter estimation problems in
mmWave systems, where $K\ll\min(M_{\textrm{vd}},M_{\textrm{hd}},M{}_{\textrm{f}},M_{\textrm{t}})$
due to the sparsity of mmWave\footnote{In rich multipath environments, i.e., $K\geq\max(M_{\textrm{vd}},M_{\textrm{hd}},M{}_{\textrm{f}})$,
no singular values and core slices of the mode-$n$ unfoldings can
be discarded, because all these belong to the signal subspace. Thus,
in this case, the tensor-based subspace estimation is equivalent to
the matrix-based counterpart \cite{tensor2HoSVD}.}.

When applying the matrix-based alternative, the noise is only suppressed
in one of the dimensions (or modes) of the measurement tensor, hence
degrading the estimation accuracy. This is because the noise is multi-dimensional
with the same dimensions as the received signal. It is important to
take all dimensions of the received signal into consideration, and
suppress the noises in all the dimensions. Thus, the use of tensors
can better suppress the noises than matrices, hence improving the
estimation accuracy of the elevation and azimuth angles and the delay,
i.e., $\hat{\theta}_{k}$, $\hat{\phi}_{k}$, and $\hat{\tau}_{k}$.
\end{remark}

\subsection{Tensor-based Spatial Smoothing for UCyA}

The parameter estimation presented in Sections V-A and V-B is actually
the last step in Fig. 1. In this subsection, we propose the necessary
optional second-to-last step. The decomposition of the signal and
noise subspaces in \eqref{eq:subspace} is under the assumption that
all the received signals are incoherent, as typically required in
the subpace-based parameter estimation algorithms, such as MUSIC \cite{34}
and ESPRIT \cite{we2}. The rank of the signal subspace is assumed
to be the number of received signals $K$. In practice, coherent signals
are often received. The rank of the signal subspace decreases, leading
to incorrect decomposition of the subspaces. An effective method to
restore the rank is a spatial smoothing technique \cite{Arraysignalprocessingconceptsandtechniques}
which divides an antenna array into several subarrays and exploits
the inherent linear recurrence relations (i.e., shift invariances)
among the subarrays to decorrelate the coherent signals. Unfortunately,
the spatial smoothing technique is only applicable to systems with
uniformly and linearly spaced antenna elements \cite{Arraysignalprocessingconceptsandtechniques}.
\begin{figure*}
\begin{centering}
\includegraphics[width=0.75\textwidth]{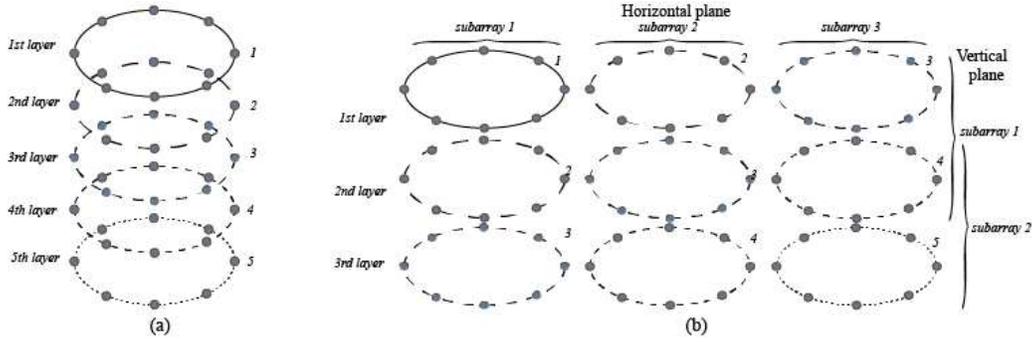}
\par\end{centering}
\caption{An illustration of the proposed spatial smoothing for a five-layer
UCyA, where we need to construct three ``subarrays'' on the horizontal
plane, the second and third UCAs are seen as the translations of the
first UCA at the same layer. After the spatial smoothing, the original
first, second, and third UCAs are at the first layer of the ``new''
UCyA, the second layer accommodates the original second, third and
fourth UCAs, and the third layer of the ``new'' UCyA accommodates
the original third, fourth and fifth UCAs. \label{fig:Spatial-smoothing-for}}
\end{figure*}

We extend the spatial smoothing technique to our hybrid UCyA to decorrelate
coherent signals. This is not trivial, as the array manifolds of the
UCyA in the horizontal space domain (i.e., the second mode of $\mathcal{Y}$)
are UCAs, not linear arrays. It is difficult to split subarrays and
obtain the required recurrence relations, as existing spatial smoothing
techniques would require. We propose to utilize the recurrence relations
between the UCAs at different layers of the UCyA to create the required
recurrence-relation subarrays in the horizontal space domain. In other
words, we regard each UCA as a subarray, and use these vertically
arranged and coaxially aligned subarrays to construct the ``virtual''
subarrays in the horizontal space domain. The $n_{\textrm{h}}$-th
subarray in the horizontal space domain can be constructed as
\begin{equation}
\mathcal{Y}_{\textrm{ss}}^{(n_{\textrm{h}})}=\mathcal{Y}\times_{1}\mathbf{J}_{\textrm{ssh,}n_{\textrm{h}}},
\end{equation}
where $\mathbf{J}_{\textrm{ssh,}n_{\textrm{h}}}=[\mathbf{0}_{M_{\textrm{hd}}\times(n_{\textrm{h}}-1)},\mathbf{I}_{M_{\textrm{hd}}},\mathbf{0}_{M_{\textrm{hd}}\times(N_{\textrm{h}}-n_{\textrm{h}})}].$

Then, we can generate the required linear recurrence relation between
two adjacent subarrays: $\mathcal{Y}_{\textrm{ss}}^{(n_{\textrm{h}}+1)}=\mathcal{Y}_{\textrm{ss}}^{(n_{\textrm{h}})}\times_{4}\mathbf{\Theta}_{\textrm{h}}$,
where $\mathbf{\Theta}_{\textrm{h}}=\textrm{diag}\left(e^{-j\frac{2\pi}{c}f_{0}h\cos(\theta_{1})},\ldots,e^{-j\frac{2\pi}{c}f_{0}h\cos(\theta_{K})}\right)$.
The numbers of subarrays and elements per subarray are determined
in the following theorem:

\begin{theorem} If both the numbers of subarrays and elements per
subarray are larger than the number of signals, i.e., $N_{\textrm{h}}\geq K$
and $M_{\textrm{hd}}\geq K$, the rank of the signal subspace in the
mode-2 of the concatenated tensor $\mathcal{Y}_{\textrm{ssh}}=\left[\underset{{\scriptstyle n_{\textrm{h}}=1,\ldots,N_{\textrm{h}}}}{\sqcup_{4}}\mathcal{Y}_{\textrm{ss}}^{(n_{\textrm{h}})}\right]$
is $K$. \end{theorem} \begin{proof} The proof can be developed
in the same way as in \cite{tensor1HoSVD}, and hence omitted. \end{proof}

According to Theorem 2, we need to construct subarrays in all domains
for the correct decomposition of the subspaces, and apply the HOSVD
in all modes of $\mathcal{Y}$. Because some of the vertically arranged
UCAs are used to construct the ``virtual'' subarrays in the horizontal
space domain, the number of subarrays in the vertical space domain
decreases. Take the five-layer UCyA in Fig. \ref{fig:Spatial-smoothing-for}(a)
for an example. The original five-layer UCyA shown in Fig. \ref{fig:Spatial-smoothing-for}(a)
becomes a three-layer virtual array, which constructs the subarrays
in the vertical space domain, as shown in Fig. \ref{fig:Spatial-smoothing-for}(b).

We propose to meticulously arrange the virtual subarrays. $N_{\textrm{v}}$
subarrays are constructed in the vertical space with $\tilde{M}_{\textrm{v}}$
elements per subarray, and $N_{\textrm{h}}$ subarrays are constructed
in the horizontal space with $M_{\textrm{hd}}$ elements per subarray.
Because there are linear recurrence relations among subcarrier frequencies,
the standard spatial smoothing technique can be used in the frequency
domain (i.e., the mode-3 of $\mathcal{Y}$). We decouple the mode-3
of $\mathcal{Y}$ into $N_{\textrm{f}}$ subarrays with $\tilde{M}_{\textrm{f}}$
elements each. As a result, the spatially smoothed tensor is given
by
\begin{align}
\mathcal{Y}_{\textrm{ss}} & =\left[\underset{{\scriptstyle n_{\textrm{v}}=1,\ldots,N_{\textrm{v}}}}{\sqcup_{4}}\underset{{\scriptstyle n_{\textrm{h}}=1,\ldots,N_{\textrm{h}}}}{\sqcup_{4}}\underset{{\scriptstyle n_{\textrm{f}}=1,\ldots,N_{\textrm{f}}}}{\sqcup_{4}}\mathcal{Y}_{\textrm{ss}}^{(n_{\textrm{v}},n_{\textrm{h}},n_{\textrm{f}})}\right]\nonumber \\
 & \in\mathbb{C}^{\tilde{M}_{\textrm{v}}\times M_{\textrm{hd}}\times\tilde{M}_{\textrm{f}}\times\left(M_{\textrm{t}}N_{\textrm{v}}N_{\textrm{h}}N_{\textrm{f}}\right)},\label{eq:spatailSmoothing}
\end{align}
which is obtained by concatenating the subarrays in mode-4:
\begin{equation}
\mathcal{Y}_{\textrm{ss}}^{(n_{\textrm{v}},n_{\textrm{h}},n_{\textrm{f}})}=\mathcal{Y}\times_{1}\mathbf{J}_{\textrm{ssvh,}n_{\textrm{vh}}}\times_{3}\mathbf{J}_{\textrm{ssf},n_{\textrm{f}}},
\end{equation}
where $n_{\textrm{vh}}=n_{\textrm{v}}+n_{\textrm{h}}-1$. $\mathbf{J}_{\textrm{ssvh,}n_{\textrm{vh}}}$
and $\mathbf{J}_{\textrm{ssf,}n_{\textrm{f}}}$ are two subtensor
selection matrices, as given respectively by

\[
\mathbf{J}_{\textrm{ssvh,}n_{\textrm{vh}}}=[\mathbf{0}_{\tilde{M}_{\textrm{v}}\times(n_{\textrm{vh}}-1)},\mathbf{I}_{\tilde{M}_{\textrm{v}}},\mathbf{0}_{\tilde{M}_{\textrm{v}}\times(N_{\textrm{vd}}-n_{\textrm{vh}})}];
\]
\begin{equation}
\mathbf{J}_{\textrm{ssf,}n_{\textrm{f}}}=[\mathbf{0}_{\tilde{M}_{\textrm{f}}\times(n_{\textrm{f}}-1)},\mathbf{I}_{\tilde{M}_{\textrm{f}}},\mathbf{0}_{\tilde{M}_{\textrm{f}}\times(N_{\textrm{f}}-n_{\textrm{f}})}].
\end{equation}
The number of subarray elements in the mode-1 and mode-3 can be computed
by $\tilde{M}_{\textrm{v}}=M_{\textrm{vd}}-N_{\textrm{v}}-N_{\textrm{h}}+2$
and $\tilde{M}_{\textrm{f}}=M_{\textrm{f}}-N_{\textrm{f}}+1$, respectively.
To decorrelate coherent signals in each domain, we use $\mathcal{Y}_{\textrm{ss}}$
to replace $\mathcal{Y}$ in \eqref{eq:modeltensorfinal}. The parameter
estimation of coherent signals follows the rest of the steps recorded
in the earlier part of Section V, which is the last step in Fig. 1.

Note that the proposed smoothing method is needed to guarantee that
the rank used for parameter estimation is the actual rank. If we conduct
the HOSVD based on a smaller rank (due to coherent signals) than the
actual rank, the estimation performance of the azimuth and elevation
angles, and delays would degrade. This is because when the smaller
rank is used, signal components can be incorrectly decomposed into
the noise subspace, reducing the dimensions of the constructed truncated
HOSVD model of $\mathcal{Y}_{\textrm{s}}$ in all modes. As a result,
we would not be able to correctly estimate the azimuth and elevation
angles, and delays.

Also note that by using the proposed method, the antenna apertures
in the first and third modes are reduced, as the elements in the two
modes of the original measurement tensor $\mathcal{Y}$ are used to
construct a sufficient number of subarrays according to Theorem 3.
The loss of the antenna aperture in the first mode is nearly one third.
The antenna aperture in the second mode does not change, because the
subarrays in the mode are constructed by the the spatial shift of
UCAs at the other layers. Algorithm \ref{alg:algorithmESPRIT} summarizes
the procedure of the proposed tensor-based subspace estimation algorithm.
\begin{algorithm}
\protect\caption{Tensor-based subspace estimation algorithm\label{alg:algorithmESPRIT}}

\begin{itemize}
\item \textbf{Input}: The received signals, $\mathbf{x}_{m_{\textrm{f}},m_{\textrm{t}},m_{\textrm{b}}}$
$(m_{\textrm{b}}=1,\ldots,M_{\textrm{b}},\:m_{\textrm{t}}=1,\ldots,M_{\textrm{t}},m_{\textrm{f}}=1,\ldots,M_{\textrm{f}})$,
the number of sources, $K$, and geometrical parameters of the UCyA.
\item \textbf{Output}: The estimated delay, elevation and azimuth angles,
$\hat{\tau}_{k}$, $\hat{\theta}_{k}$, and $\hat{\phi}_{k}$, $k=1,2,\ldots,K$.
\item Design the analog and digital beamforming matrices, $\mathbf{B}_{\textrm{ab}}$
and $\mathbf{B}_{\textrm{db},m_{\textrm{f}},m_{\textrm{b}}}$, and
model the beamspace signals according to \eqref{eq:tensormodel}.
\item Calculate the focusing matrices, $\mathbf{T}_{\textrm{v},m_{\textrm{f}},m_{\textrm{b}}}$
and $\mathbf{T}_{\textrm{h},m_{\textrm{f}},m_{\textrm{b}}}$, by solving
\eqref{eq:wideband1} and \eqref{eq:wideband2}, and formulate the
signals according to \eqref{COHERENT}.
\item Collect all the sweeping results in \eqref{eq:ELEMENTCOHERENT} and
formulate them as $\mathcal{Y}$.
\item Construct the spatially smoothed tensor $\mathcal{Y}_{\textrm{ss}}$
by using \eqref{eq:spatailSmoothing}.
\item Take HOSVD of $\mathcal{Y}_{\textrm{ss}}$ and get $\mathcal{U}_{\textrm{s}}$
according to \eqref{eq:truncatedHOSVD} and \eqref{eq:Us}.
\item Use TLS-ESPRIT \eqref{eq:UsJv}-\eqref{eq:TLSfinal}, and estimate
$\hat{\theta}_{k}$ and $\tau_{k}$ by using \eqref{theta} and \eqref{delay},
respectively.
\item Calculate the noise subspace tensor $\mathcal{Y}_{\textrm{n}}$ in
\eqref{eq:noisesubspace} and estimate $\hat{\phi}_{k}$ by searching
the prominent peaks of \eqref{eq:MUSIC2}.
\end{itemize}
\end{algorithm}

\subsection{Complexity Analysis}

The hardware and software complexity of the proposed tensor-based
parameter estimation algorithm is analyzed. The proposed hybrid beamformers
reduces the hardware complexity to $O(M_{\textrm{bsr}})=O(PM_{\textrm{v}})$,
while fully digital beamformers using the same number of antennas
have hardware complexity $O(M_{\textrm{bs}})$.

As for signal processing complexity, we compare the computational
complexity of the proposed tensor-based algorithm with its matrix-based
counterpart and the state-of-the-art CP-based orthogonal matching
pursuit (CP-OMP) algorithm. For matrix-based algorithms, the computational
complexity of performing SVD on the measurement sample matrix and
truncating its rank to $K$ is $O(PM_{\textrm{v}}M_{\textrm{f}}M_{\textrm{t}}K)$.
The complexities of estimating the delay, elevation angle, and azimuth
angle are $O(K^{3}+PM_{\textrm{v}}M_{\textrm{f}})$, $O(K^{3}+PM_{\textrm{v}})$,
and $O(PK^{2}+P^{2}KD),$ respectively. $D$ is the size of search
dimension. Thus, the overall complexity of the matrix-based estimation
is $O(PM_{\textrm{v}}M_{\textrm{f}}M_{\textrm{t}}K+PM_{\textrm{v}}M_{\textrm{f}}+K^{3}+PM_{\textrm{v}}+PK^{2}+P^{2}KD)$.
For the proposed tensor-based algorithm, the truncated HOSVD of the
measurement tensor evaluates the SVD of its matricized form in each
mode and discards insignificant singular vectors. The complexity is
$O(4PM_{\textrm{v}}M_{\textrm{f}}M_{\textrm{t}}K)=O(PM_{\textrm{v}}M_{\textrm{f}}M_{\textrm{t}}K)$.
The complexity of computing the core $\mathcal{L}_{\textrm{n}}$ and
the tensor signal subspace $\mathcal{U}_{\textrm{s}}$ in \eqref{eq:Us}
is $O(PM_{\textrm{v}}M_{\textrm{f}}M_{\textrm{t}}K+PM_{\textrm{v}}M_{\textrm{f}}K^{2})$.
The complexities of estimating delay, elevation and azimuth angles
are $O(PM_{\textrm{v}}M_{\textrm{f}}+K^{3})$, $O(PM_{\textrm{v}}M_{\textrm{f}}+K^{3})$
and $O(PM_{\textrm{v}}M_{\textrm{f}}M_{\textrm{t}}K+P^{2}KD)$, respectively.
The tensor-based algorithm needs slightly more computations, but its
complexity is still in the same order with that of its matrix-based
counterpart. The CP-OMP algorithm \cite{tensormmWave} applies CP
decomposition to decompose the received signal tensor model, and then
uses OMP to estimate the parameters. The complexities are $O(PM_{\textrm{v}}M_{\textrm{f}}M_{\textrm{t}}K+PM_{\textrm{v}}M_{\textrm{f}}K^{2}+K^{4})$
and $O(PM_{\textrm{v}}M_{\textrm{f}}M_{\textrm{t}}(N_{1}+N_{2}+N_{3}+N_{4}))$,
respectively, where $N_{1}\gg K$, $N_{2}\gg K$, $N_{3}\gg K$ and
$N_{4}\gg K$ are the dimensions of the OMP grid. The CP-OMP algorithm
has a much higher complexity than that of our algorithm. A comparison
study of computational complexity between the three algorithms is
provided in Table 1, which also shows the computer runtime obtained
by running MATLAB simulations on a ThinkPad X1 Carbon with an i5 processor
and 8 GB memory (where $P=12$, $M_{\textrm{v}}=20$, $M_{\textrm{t}}=20$,
$M_{\textrm{f}}=20$, $K=5$, $D=50$ and $N_{i}=50$ for $i=1,2,3,4$).

\begin{table*}
\centering{}\caption{Computational complexity and CPU running time of three algorithms.}
\begin{tabular}{|>{\centering}p{2cm}|>{\centering}p{4.2cm}|>{\centering}p{4.2cm}|>{\centering}p{4.2cm}|}
\hline
 & Proposed tensor-based algorithm  & Matrix-based counterpart  & CP-OMP algorithm\tabularnewline
\hline
\hline
Channel decomposition  & $O(PM_{\textrm{v}}M_{\textrm{f}}M_{\textrm{t}}K)$ (385.79 ms)  & $O(PM_{\textrm{v}}M_{\textrm{f}}M_{\textrm{t}}K)$ (102.58 ms)  & $O(PM_{\textrm{v}}M_{\textrm{f}}M_{\textrm{t}}K+K^{4}$$+PM_{\textrm{v}}M_{\textrm{f}}K^{2})$
(1284.82 ms)\tabularnewline
\hline
Parameter estimation  & $O(PM_{\textrm{v}}M_{\textrm{f}}K^{2}+PM_{\textrm{v}}M_{\textrm{f}}M_{\textrm{t}}K+K^{3}+P^{2}KD)$
(189.74 ms)  & $O(PM_{\textrm{v}}M_{\textrm{f}}+K^{3}+PM_{\textrm{v}}+PK^{2}+P^{2}KD)$
(75.09 ms)  & $O(PM_{\textrm{v}}M_{\textrm{f}}M_{\textrm{t}}(N_{1}+N_{2}+N_{3}+N_{4}))$
(1740.59 ms) \tabularnewline
\hline
Total  & $O(PM_{\textrm{v}}M_{\textrm{f}}+K^{3}+PM_{\textrm{v}}M_{\textrm{f}}M_{\textrm{t}}K+PM_{\textrm{v}}M_{\textrm{f}}K^{2}+P^{2}KD)$ (575.52 ms)
& $O(PM_{\textrm{v}}M_{\textrm{f}}M_{\textrm{t}}K+PM_{\textrm{v}}M_{\textrm{f}}+K^{3}+PM_{\textrm{v}}+PK^{2}+P^{2}KD)$

(177.68 ms)  & $O(PM_{\textrm{v}}M_{\textrm{f}}M_{\textrm{t}}K+PM_{\textrm{v}}M_{\textrm{f}}K^{2}+K^{4}+PM_{\textrm{v}}M_{\textrm{f}}M_{\textrm{t}}(N_{1}+N_{2}+N_{3}+N_{4}))$ (3025.41 ms) \tabularnewline
\hline
\end{tabular}
\end{table*}

\section{Simulation Results}

In this section, simulation results are provided to to demonstrate
the performance of the proposed algorithm. We simulate a system with
2 GHz bandwidth and a total of 2,000 subcarriers. Out of the total
2,000 subcarriers, $M_{\textrm{f}}=20$ evenly spaced subcarriers
are selected for the proposed channel parameter estimation. Each of
the subcarriers undergoes flat fading. The reference frequency $f_{0}=28$
GHz, and the number of time frames is $M_{\textrm{t}}$= 20. To evaluate
the performance of the proposed algorithm in typical mmWave channels,
all the channel parameters are set according to 3GPP TR 38.901 \cite{3GPP}.
An Urban Micro (UMi) scenario is considered in our simulation, and
thus, the UMi pathloss model presented in \cite{3GPP} is applied.
 The number of time frames is set
to $M_{\textrm{t}}$= 20. We assume that there are $K=5$ signals,
two of which are coherent. The actual azimuth angles, elevation angles,
and delays of the signals are set up randomly each time. The distance
between vertically adjacent UCAs is $h=0.5\lambda_{0}$ and the radius
of the UCyA is $r=2\lambda_{0}$, where $\lambda_{0}=c/f_{0}$.

We compare the proposed tensor-based coherent TLS (T-CTLS) algorithm
with its variation (T-CLS) which applies the LS procedure for solving
the invariance equation \eqref{eq:Usv12}; its variation without using
smoothing (T-CTLS w/o S); its reduced version in the matrix form (M-CTLS);
the state-of-the-art matrix-based incoherent generalized beamspace
MUSIC (M-IGBM) \cite{we11}; the tensor-based incoherent MUSIC (T-IM)
\cite{tensorMUSIC}; and the state-of-the-art CP-OMP \cite{tensormmWave}.
The CRLB is derived according to \cite{***ding9,bounds}. Note that
both CP-OMP and our proposed parameter estimation algorithms are only
applicable for additive Gaussian noises, where the noises are independent
between different antennas and the noise power is identical at the
antennas. This is because the algorithms which exploit the second-order
statistics of the received signals cannot correctly decompose the
signal and (non-Gaussian) noise subspaces.

\begin{figure*}
\begin{centering}
\includegraphics[width=15cm]{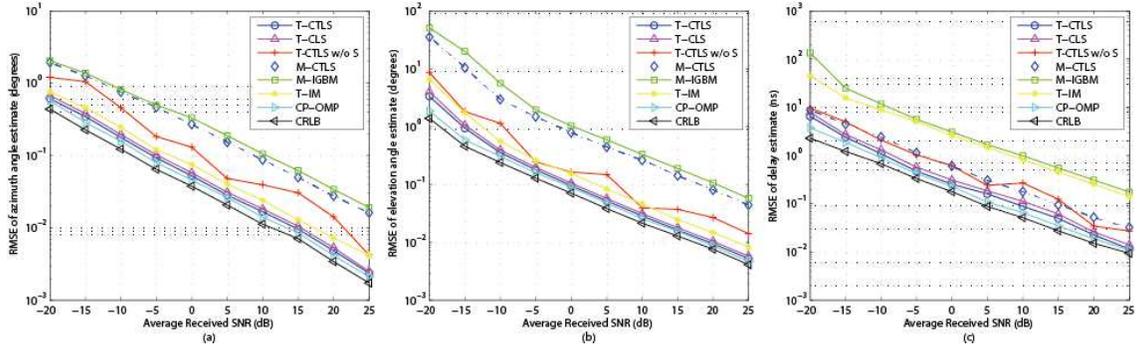}
\par\end{centering}
\caption{RMSE vs. the average received SNR for the estimation of different
parameters. (a) Azimuth angle; (b) Elevation angle; (c) Delay ($f_{0}=28$
GHz; $B=2$ GHz; $M_{\textrm{t}}$= 20; $M_{\textrm{f}}=20$; $K=5$;
and $M_{\textrm{bs}}=400$).\label{fig:snr}}
\end{figure*}

\begin{figure*}
\begin{centering}
\includegraphics[width=15cm]{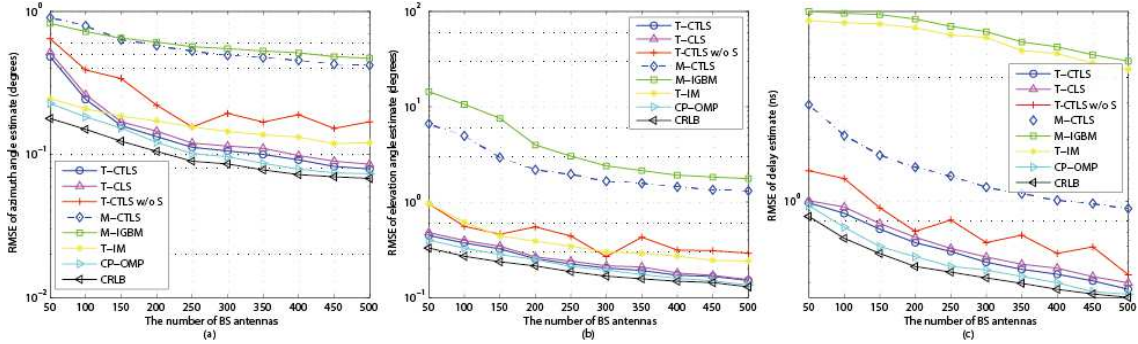}
\par\end{centering}
\caption{RMSE vs. the number of BS antennas for the estimation of different
parameters. (a) Azimuth angle; (b) Elevation angle; (c) Delay ($f_{0}=28$
GHz; $B=2$ GHz; $M_{\textrm{t}}$= 20; $M_{\textrm{f}}=20$; $K=5$;
and SNR$=$-5 dB).\label{fig:antenna}}
\end{figure*}

Fig. \ref{fig:snr} plots the root mean square errors (RMSEs) for
the estimates of azimuth angles, elevation angles, and delays of the
signals versus the average received SNR, where the BS has 400 receive
antennas. Fig. \ref{fig:snr} shows that our proposed T-CTLS algorithm
outperforms the other algorithms, and its RMSE approaches the CRLB.
In Figs. \ref{fig:snr}(a) and (b), we see that the tensor-based algorithms
provide higher accuracy than their matrix-based counterparts, especially
in low SNR regimes. The matrix-based algorithms are less robust to
noises than the proposed tensor-based algorithms. We also see that
CP-OMP has slightly better performance than our proposed algorithm,
due to the fact that CP decomposition can be regarded as a maximum
likelihood method under the additive Gaussian noise. However, its
performance improvement is limited since OMP can only generate discrete
estimates. In addition, CP-OMP also has a much higher complexity than
our algorithm, as analyzed in Section V-D. Fig. \ref{fig:snr}(c)
shows that the methods applying coherent wideband signal preprocessing
outperform those employing incoherent wideband preprocessing, in terms
of delay estimation, because the former fully exploits the high temporal
resolution offered by wideband mmWave systems.

Fig. \ref{fig:antenna} shows the RMSEs versus the number of receive
antennas under -5 dB SNR. It is seen that the RMSE of the estimated
parameters approaches the CRLB, as the number of antennas increases.
However, when the number of antennas is not very large, e.g. less
than 100, the algorithms, including T-CTLS, T-CTL, and M-CTLS, cannot
achieve accurate azimuth angle estimation, as shown in Fig. \ref{fig:antenna}(a).
The reason is that the conditions of Theorem 1 may not be met, and
thus the approximation in \eqref{eq:Theorem1} becomes inaccurate.
Nevertheless, when the number of antennas is large, the RMSEs of these
three algorithms decrease fast, and T-CTLS rapidly outperforms the
others. By comparing Figs. \ref{fig:snr} and \ref{fig:antenna},
we also see that if the proposed spatial smoothing technique is not
applied, the estimation accuracy of the proposed algorithm decreases
noticeably. This is because two coherent signals are decorrelated,
the signal and noise subspaces can be incorrectly decoupled without
spatial smoothing, and the parameters of the coherent signals cannot
be precisely estimated.

In order to validate Theorem 1, Fig. \ref{fig:Pp} plots the RMSE
of the parameter estimation versus the highest order, $P$, with different
numbers of horizontal array steering vectors. The SNR is -5 dB. We
see that when $P$ is less than 10 or the number of the horizontal
array steering vectors in \eqref{eq:horizonarray} is 20, the algorithms
applying Theorem 1 to design the hybrid beamformers (i.e., T-CTLS
and M-CTLS), cannot achieve satisfactory estimation, because the number
of the transformed beamspace vectors \eqref{eq:Theorem1} is not sufficient
to represent the array response vectors. When $P\geq12$, regardless
of the number of array response vectors, increasing the beamspace
vectors has little impact on the estimation. By exploiting this property,
we can reduce the number of required RF chains and, in turn, the hardware
cost.

Fig. \ref{fig:Mfvs} shows the RMSE of the estimated azimuth angles,
elevation angles, and delays, with an increasing number of received
paths. T-CTLS and M-CTLS are tested. We set SNR to -5 dB and $M_{\textrm{f}}=8$.
We observe that the performance gap between the matrix and tensor
forms of the proposed algorithm, i.e., M-CTLS and T-CTLS, decreases
with the increasing number of received paths. This is because the
noise components which can be suppressed by using the tensor-based
algorithms in the first, second, and third modes of $\mathcal{Y}$,
depend on the difference between the number of paths and the tensor
dimension in each mode of $\mathcal{Y}$. As the number of received
paths increases, the gain of the tensor-based algorithm, T-CTLS, diminishes.
The performance gap remains consistent in Fig. \ref{fig:Mfvs}(a)
though. This is because, despite the number of paths increases, the
dimension in the first mode of $\mathcal{Y}$, i.e., $M_{\textrm{vd}}=2P+1$,
is still much larger than the number of paths. Moreover, we estimate
the azimuth angles with tensor-MUSIC in \eqref{eq:music}. The method
involves peak search, and is hardly affected by the number of paths.
In conclusion, the new tensor-based algorithm, T-CTLS, can achieve
much better performance than its matrix-based counterpart, especially
under B5G settings where the number of received paths is small due
to the sparsity of mmWave propagation.

\begin{figure*}
\begin{centering}
\includegraphics[width=15cm]{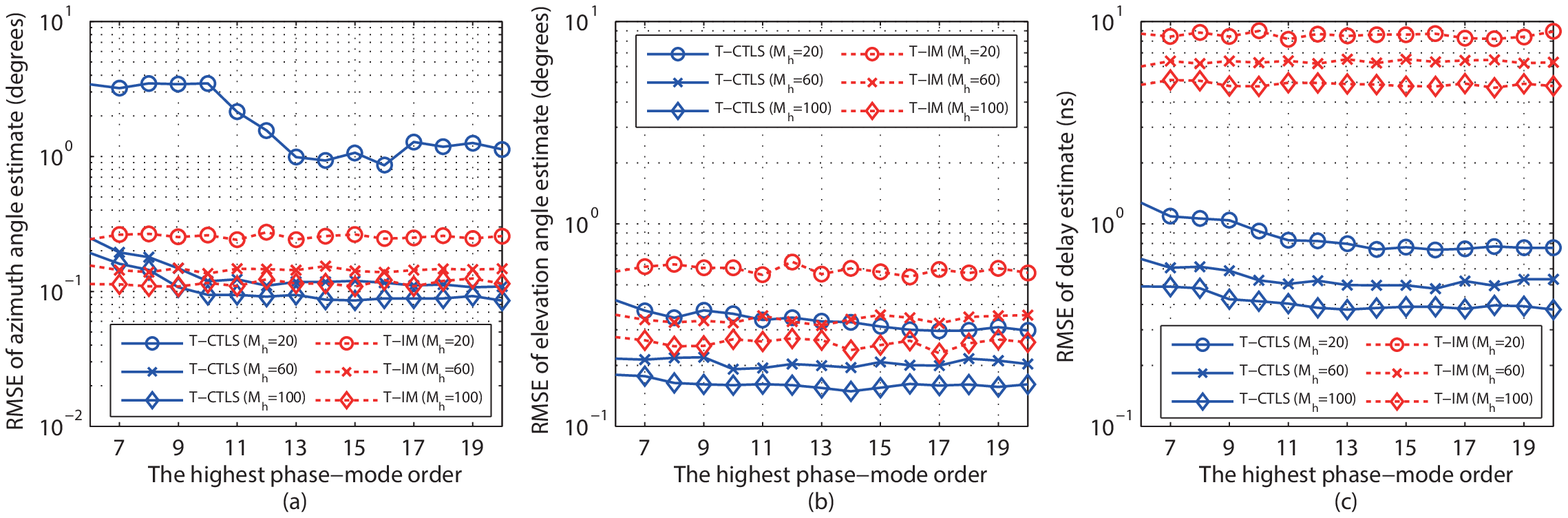}
\par\end{centering}
\caption{RMSE vs. the highest beamspace dimension. (a) Azimuth angle; (b) Elevation
angle; (c) Delay ($f_{0}=28$ GHz; $B=2$ GHz; $M_{\textrm{t}}$=
20; $M_{\textrm{f}}=20$; $K=5$; and SNR$=$-5 dB).\label{fig:Pp}}
\end{figure*}

\begin{figure*}
\begin{centering}
\includegraphics[width=15cm]{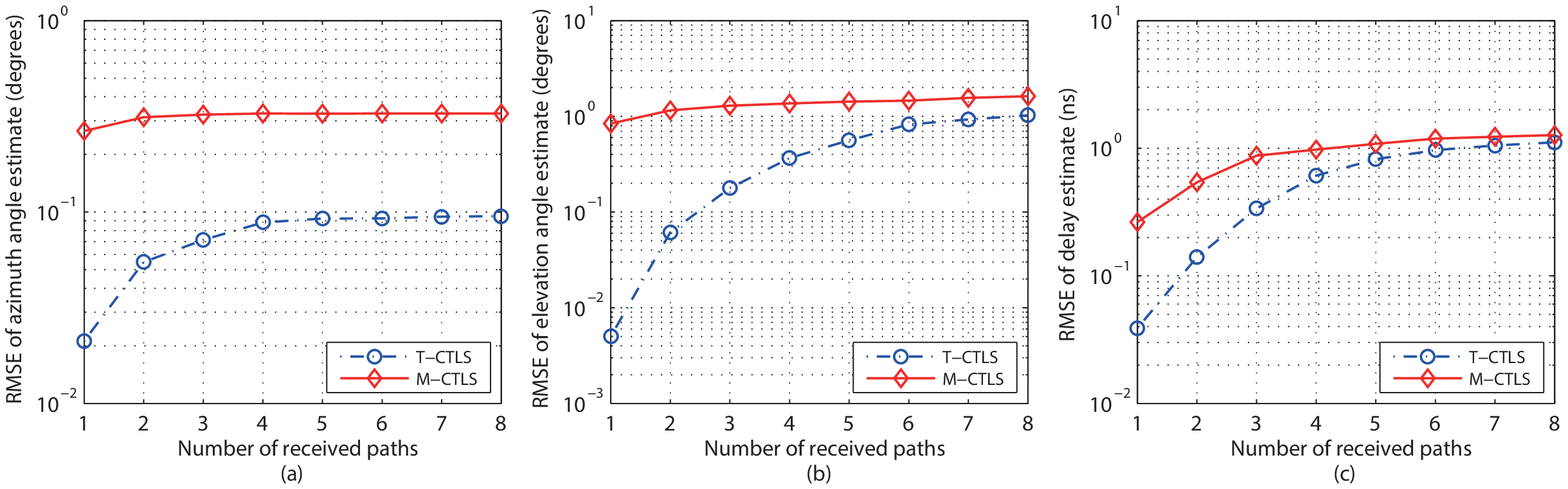}
\par\end{centering}
\caption{RMSE vs. the number of received paths. (a) Azimuth angle; (b) Elevation
angle; (c) Delay ($f_{0}=28$ GHz; $B=2$ GHz; $M_{\textrm{t}}$=
20; $M_{\textrm{f}}=8$; and SNR$=$-5 dB).\label{fig:Mfvs}}
\end{figure*}

\section{Conclusion}

We presented a new tensor-based multi-dimensional channel parameter
estimation algorithm for 5G/B5G wideband mmWave large-scale hybrid
antenna arrays. By exploiting the multidimensional structure of the
received signals, the algorithm can suppress the noises across all
domains of the received signals, improving estimation accuracy. Specifically,
we designed the hybrid beamformers to maintain the angular resolution
and suppress beam squinting. We developed the new HOSVD model to suppress
the noise, and revealed the shift-invariance relations in the tensor
form. Given the relations, we designed the new tensor-based TLS-ESPRIT
algorithm for parameter estimation. We also rearranged the measurement
tensor to estimate coherent signals. By applying the channel parameters
presented by 3GPP TR 38.901 {[}31{]}, simulations show that the proposed
tensor-based algorithm can accurately estimate the multi-dimensional
parameters in typically used mmWave channels, even under low SNRs.

\section*{Appendix I\protect \protect \protect \protect \protect \protect
\protect \protect }

\section*{Properties of Tensor Operation}

The important properties of tensor operations used in this paper are
provided below.

\begin{property} The $n$-mode product satisfies the following properties:
\begin{equation}
\mathcal{A}\times_{n}\mathbf{B}\times_{n}\mathbf{C}=\mathcal{A}\times_{n}\left(\mathbf{CB}\right);\label{eq:property1}
\end{equation}
\begin{equation}
\mathcal{A}\times_{n}\mathbf{B}\times_{m}\mathbf{D}=\mathcal{A}\times_{m}\mathbf{D}\times_{n}\mathbf{B},\label{eq:property2}
\end{equation}
where $\mathcal{A}\in\mathbb{C}^{I_{1}\times I_{2}\times\cdots\times I_{N}}$
, $\mathbf{B}\in\mathbb{C}^{J_{n}\times I_{n}}$, $\mathbf{C}\in\mathbb{C}^{K_{n}\times J_{n}}$,
and $\mathbf{D}\in\mathbb{C}^{J_{m}\times I_{m}}$ $(n,m=1,2,\ldots,N$
and $n\neq m)$. \end{property}

\begin{property} The Tucker decomposition decomposes a tensor $\mathcal{A}\in\mathbb{C}^{I_{1}\times I_{2}\times\cdots\times I_{N}}$
into a core tensor $\mathcal{G}\in\mathbb{C}^{R_{1}\times R_{2}\times\cdots\times R_{N}}$
multiplied by a factor matrix $\mathbf{C}^{(n)}=\left[\mathbf{c}_{r_{n}=1}^{(n)},\mathbf{c}_{r_{n}=2}^{(n)},\ldots,\mathbf{c}_{r_{n}=R_{n}}^{(n)}\right]\in\mathbb{C}^{I_{n}\times R_{n}}$
($\mathbf{c}_{r_{n}}^{(n)}\in\mathbb{C}^{I_{n}\times1}$ and $n=1,2,\ldots,N$)
in each mode, i.e.,
\begin{align}
\mathcal{A} & =\sum_{r_{1}=1}^{R_{1}}\sum_{r_{2}=1}^{R_{2}}\cdots\sum_{r_{N}=1}^{R_{N}}g_{r_{1}r_{2}\cdots r_{N}}\left(\mathbf{c}_{r_{1}}^{(1)}\circ\mathbf{c}_{r_{2}}^{(2)}\circ\cdots\mathbf{c}_{r_{N}}^{(N)}\right)\nonumber \\
 & =\left\llbracket \mathcal{G};\mathbf{C}^{(1)},\mathbf{C}^{(2)},\ldots,\mathbf{C}^{(N)}\right\rrbracket .\label{eq:Tucker}
\end{align}
The HOSVD is a special
case of the Tucker decomposition, where the core tensor is all-orthogonal
\cite{TensorDecompositions1}, and the factor matrices are the unitary
left singular matrices of the mode-$n$ unfolding of $\mathcal{A}$.
\end{property}

\begin{property} The CANDECOMP/PARAFAC (CP) decomposition decomposes
a tensor $\mathcal{A}\in\mathbb{C}^{I_{1}\times I_{2}\times\cdots\times I_{N}}$
into a sum of rank-one component tensors $\mathbf{b}_{r}^{(n)}\in\mathbb{C}^{I_{n}}$,
as given by
\begin{equation}
\mathcal{A}=\sum_{r=1}^{R}\lambda_{r}\mathbf{b}_{r}^{(1)}\circ\mathbf{b}_{r}^{(2)}\circ\cdots\mathbf{b}_{r}^{(N)},\label{eq:cpdefine}
\end{equation}
where $R=\textrm{Rank}(\mathcal{A})$ is the rank of $\mathcal{A}$\footnote{The rank of a tensor, $\mathcal{A},$ denoted $\textrm{Rank}(\mathcal{A})$,
is defined as the smallest number of rank-one tensors that yield $\mathcal{A}$
in a linear combination \cite{TensorDecompositions1}.}. Following \cite{TensorDecompositions1}, CP can be viewed as the
special case of the Tucker decomposition, where the core tensor is
superdiagonal. Thus, the CP model in \eqref{eq:cpdefine} can be rewritten
as a multilinear product:
\begin{align}
\mathcal{A} & =\mathcal{D}\times_{1}\mathbf{B}^{(1)}\times_{2}\mathbf{B}^{(2)}\cdots\times_{N}\mathbf{B}^{(N)}\nonumber \\
 & =\left\llbracket \mathcal{D};\mathbf{B}^{(1)},\mathbf{B}^{(2)},\ldots,\mathbf{B}^{(N)}\right\rrbracket ,\label{eq:CP2}
\end{align}
where $\mathbf{B}^{(n)}=\left[\mathbf{b}_{1}^{(n)},\mathbf{b}_{2}^{(n)},\ldots,\mathbf{b}_{R}^{(n)}\right]\in\mathbb{C}^{J_{n}\times R}$
is the factor matrix of $\mathbf{b}_{r}^{(n)}$, and $\mathcal{D}\in\mathbb{C}^{R\times R\times\cdots\times R}$
is a superdiagonal tensor\footnote{A tensor $\mathcal{A}\in\mathbb{C}^{I_{1}\times I_{2}\times\cdots\times I_{N}}$
is diagonal if $a_{i_{1}i_{2}\cdots i_{N}}\neq0$ only if $i_{1}=i_{2}=\cdots=i_{N}$.
When $I_{1}=I_{2}=\cdots=I_{N}$, $\mathcal{A}$ is called as superdiagonal.} with $d_{r,r,\cdots,r}=\lambda_{r}$. \end{property}

\begin{property} The multilinear product of a tensor $\mathcal{A}\in\mathbb{C}^{I_{1}\times I_{2}\times\cdots\times I_{N}}$
with matrices $\mathbf{B}^{(n)}\in\mathbb{C}^{J_{n}\times I_{n}}$,
$n=1,2,\ldots,N$, is a sequence of contractions, each being an $n$-mode
product, i.e.,
\begin{equation}
\mathcal{C}=\mathcal{A}\times_{1}\mathbf{B}^{(1)}\times_{2}\mathbf{B}^{(2)}\cdots\times_{N}\mathbf{B}^{(N)}\in\mathbb{C}^{J_{1}\times J_{2}\times\cdots\times J_{N}}.\label{eq:multiproduct}
\end{equation}
Its mode-$n$ unfolding is given by
\begin{align}
\mathbf{C}_{(n)} & =\mathbf{B}^{(n)}\mathbf{A}_{(n)}(\mathbf{B}^{(n+1)}\otimes\mathbf{B}^{(n+2)}\otimes\nonumber \\
 & \cdots\otimes\mathbf{B}^{(N)}\otimes\mathbf{B}^{(1)}\otimes\mathbf{B}^{(2)}\otimes\cdots\otimes\mathbf{B}^{(n-1)})^{T}.\label{eq:propertymultiproduct}
\end{align}
\end{property}

\section*{Appendix II\protect \protect \protect \protect \protect \protect
\protect \protect }

\section*{Proof of Theorem 1}

Let $\gamma_{m_{\textrm{f}}}(\theta_{k_{m_{\textrm{b}}}})=\frac{2\pi}{c}f_{m_{\textrm{f}}}r\sin(\theta_{k_{m_{\textrm{b}}}}).$
The Q-DFT of $a_{\textrm{h},m_{\textrm{f}},m_{\textrm{b}}}(\theta_{k_{m_{\textrm{b}}}},\phi_{k_{m_{\textrm{b}}}})$
can be expressed as
\begin{align}
 & a_{\textrm{QDFT},p,m_{\textrm{f}},m_{\textrm{b}}}(\theta_{k_{m_{\textrm{b}}}},\phi_{k_{m_{\textrm{b}}}})\nonumber \\
 & \stackrel{(\textrm{a})}{=}\sum_{m_{\textrm{h}}=1}^{M_{\textrm{h}}}\left(\frac{1}{\sqrt{M_{\textrm{h}}}}\sum_{q=-\infty}^{\infty}j^{q}J_{q}(\gamma_{m_{\textrm{f}}}(\theta_{k_{m_{\textrm{b}}}}))e^{jq(\phi_{k_{m_{\textrm{b}}}}-\varphi_{m_{\textrm{h}}})}\right)\nonumber \\
 & \qquad\qquad\times e^{-j\frac{2\pi(m_{\textrm{h}}-1)}{M_{\textrm{h}}}p}\nonumber \\
 & \stackrel{(\textrm{b})}{=}\frac{1}{\sqrt{M_{\textrm{h}}}}\sum_{Q=-\infty}^{\infty}M_{\textrm{h}}j^{(QM_{\textrm{h}}-p)}J_{(QM_{\textrm{h}}-p)}(\gamma_{m_{\textrm{f}}}(\theta_{k_{m_{\textrm{b}}}}))\nonumber \\
 & \qquad\qquad\times e^{j(QM_{\textrm{h}}-p)\phi_{k_{m_{\textrm{b}}}}}\nonumber \\
 & \stackrel{(\textrm{c})}{=}\sqrt{M_{\textrm{h}}}[j^{p}J_{p}(\gamma_{m_{\textrm{f}}}(\theta_{k_{m_{\textrm{b}}}}))e^{-jp\phi_{k_{m_{\textrm{b}}}}}\nonumber \\
 & +\sum_{Q=-\infty,Q\neq0}^{\infty}\varepsilon_{p,Q}(\gamma_{m_{\textrm{f}}}(\theta_{k_{m_{\textrm{b}}}}),\phi_{k_{m_{\textrm{b}}}})]\label{eq:T1}
\end{align}
where
\begin{align}
 & \varepsilon_{p,Q}(\gamma_{m_{\textrm{f}}}(\theta_{k_{m_{\textrm{b}}}}),\phi_{k_{m_{\textrm{b}}}})\nonumber \\
 & =j^{(QM_{\textrm{h}}-p)}J_{(QM_{\textrm{h}}-p)}(\gamma_{m_{\textrm{f}}}(\theta_{k_{m_{\textrm{b}}}}))e^{j(QM_{\textrm{h}}-p)\phi_{k_{m_{\textrm{b}}}}}.
\end{align}
In \eqref{eq:T1}, $(\textrm{a})$ and $(\textrm{c})$ follow the
important properties of the Bessel function, i.e., $e^{jx\cos y}=\sum_{v=-\infty}^{\infty}j^{v}J_{v}(x)e^{jvy}$
and $J_{-v}(x)=(-1)^{v}J_{v}(x)$, respectively; $(\textrm{b})$ is
obtained by letting $p+q=QM_{\textrm{h}}$; and $(\textrm{c})$ stems
from the property of the Bessel function $J_{-v}(x)=(-1)^{v}J_{v}(x)$
\cite{bessel}.

Consider that the number of antennas per UCA, $M_{\textrm{h}}$, is
large, i.e, $M_{\textrm{h}}\gg P$. Let $M_{\textrm{h}}=\alpha P$
and $\gamma_{m_{\textrm{f}}}(\theta_{k_{m_{\textrm{b}}}})=\beta P,$
where $\alpha\gg1$ and $0<\beta<1$. According to \cite{bessel},
we have $J_{v}(v\rho)<J_{v}(v)$ and $J_{v_{1}}(v_{1}\rho)<J_{v_{2}}(v_{2}\rho)$,
where $v_{1}>v_{2}$ and $\rho\in(0,1)$. Since $P\geq\left\lfloor 2\pi f_{m_{\textrm{f}}}r/c\right\rfloor $,
we have $J_{(QM_{\textrm{h}}-p)}(\gamma_{m_{\textrm{f}}}(\theta_{k_{m_{\textrm{b}}}}))<J_{(\alpha-1)P}\left(\beta P\right)$
and $J_{P}\left(\beta P\right)\leq J_{p}(\gamma_{m_{\textrm{f}}}(\theta_{k_{m_{\textrm{b}}}}))$.
Set $\alpha=3$ and $\beta=0.5$ for an example. In general, $P>3$.
Hence,$J_{p}(\gamma_{m_{\textrm{f}}}(\theta_{k_{m_{\textrm{b}}}}))\geq J_{3}\left(1.5\right)\approx0.06$
and
\begin{align}
 & J_{(Q_{2}M_{\textrm{h}}-p)}(\gamma_{m_{\textrm{f}}}(\theta_{k_{m_{\textrm{b}}}}))<J_{(Q_{1}M_{\textrm{h}}-p)}(\gamma_{m_{\textrm{f}}}(\theta_{k_{m_{\textrm{b}}}}))\nonumber \\
 & <J_{(M_{\textrm{h}}-p)}(\gamma_{m_{\textrm{f}}}(\theta_{k_{m_{\textrm{b}}}}))<J_{6}\left(1.5\right)\approx0.0002,
\end{align}
where $Q_{2}>Q_{1}>1$. Compared with $J_{p}(\gamma_{m_{\textrm{f}}}(\theta_{k_{m_{\textrm{b}}}}))$,
the amplitude of $J_{(QM_{\textrm{h}}-p)}(\gamma_{m_{\textrm{f}}}(\theta_{k_{m_{\textrm{b}}}}))$
is so small and can be omitted. We suppress $\varepsilon_{p,Q}(\gamma_{m_{\textrm{f}}}(\theta_{k_{m_{\textrm{b}}}}),\phi_{k_{m_{\textrm{b}}}})$
and approximate \eqref{eq:T1} as
\begin{align}
 & a_{\textrm{QDFT},p,m_{\textrm{f}},m_{\textrm{b}}}(\theta_{k_{m_{\textrm{b}}}},\phi_{k_{m_{\textrm{b}}}})\nonumber \\
 & \approx\sqrt{M_{\textrm{h}}}j^{p}J_{p}(\gamma_{m_{\textrm{f}}}(\theta_{k_{m_{\textrm{b}}}}))\exp(-jp\phi_{k_{m_{\textrm{b}}}}).
\end{align}
This concludes the proof of Theorem 1.

 \bibliographystyle{IEEEbib}
\bibliographystyle{IEEEbib}

\begin{IEEEbiography}[{\includegraphics[width=1in,height=1.25in,clip,keepaspectratio]{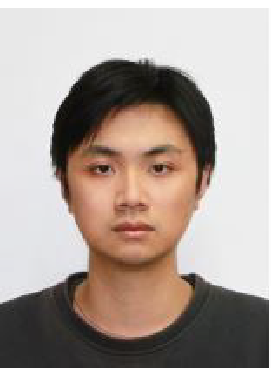}}]{Zhipeng Lin}

(M'20)  is currently working toward the dual Ph.D. degrees in communication and information engineering with the School of Information and Communication Engineering, Beijing University of Posts and Telecommunications, Beijing, China, and the School of Electrical and Data Engineering, University of Technology of Sydney,  NSW, Australia. His current research interests include millimeter-wave communication, massive MIMO, hybrid beamforming, wireless localization, and tensor processing.

\end{IEEEbiography}

\begin{IEEEbiography}[{\includegraphics[width=1in,height=1.25in,clip,keepaspectratio]{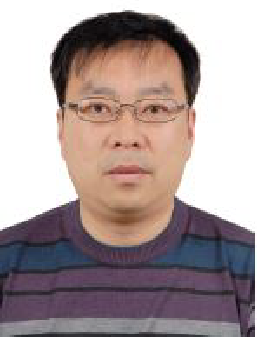}}]{Tiejun Lv}
(M'08-SM'12) received the M.S. and Ph.D. degrees in electronic engineering from the University of Electronic Science and Technology of China (UESTC), Chengdu, China, in 1997 and 2000, respectively. From January 2001 to January 2003, he was a Postdoctoral Fellow with Tsinghua University, Beijing, China. In 2005, he was promoted to a Full Professor with the School of Information and Communication Engineering, Beijing University of Posts and Telecommunications (BUPT). From September 2008 to March 2009, he was a Visiting Professor with the Department of Electrical Engineering, Stanford University, Stanford, CA, USA. He is the author of 3 books, more than 80 published IEEE journal papers and 180 conference papers on the physical layer of wireless mobile communications. His current research interests include signal processing, communications theory and networking. He was the recipient of the Program for New Century Excellent Talents in University Award from the Ministry of Education, China, in 2006. He received the Nature Science Award in the Ministry of Education of China for the hierarchical cooperative communication theory and technologies in 2015.
\end{IEEEbiography}

\begin{IEEEbiography}[{\includegraphics[width=1in,height =1.25in,clip,keepaspectratio]{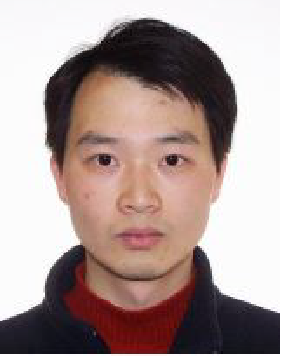}}]{Wei Ni}
(M'09-SM'15) received the B.E. and Ph.D. degrees in Electronic Engineering from Fudan University, Shanghai, China, in 2000 and 2005, respectively. Currently, he is a Group Leader and Principal Research Scientist at CSIRO, Sydney, Australia, and an Adjunct Professor at the University of Technology Sydney and Honorary Professor at Macquarie University, Sydney. He was a Postdoctoral Research Fellow at Shanghai Jiaotong University from 2005 -- 2008; Deputy Project Manager at the Bell Labs, Alcatel/Alcatel-Lucent from 2005 to 2008; and Senior Researcher at Devices R\&D, Nokia from 2008 to 2009. His research interests include signal processing, stochastic optimization, learning, as well as their applications to network efficiency and integrity.

Dr Ni is the Chair of IEEE Vehicular Technology Society (VTS) New South Wales (NSW) Chapter since 2020 and an Editor of IEEE Transactions on Wireless Communications since 2018. He served first the Secretary and then Vice-Chair of IEEE NSW VTS Chapter from  2015 to 2019, Track Chair for VTC-Spring 2017, Track Co-chair for IEEE VTC-Spring 2016, Publication Chair for BodyNet 2015, and  Student Travel Grant Chair for WPMC 2014.
\end{IEEEbiography}

\begin{IEEEbiography}[{\includegraphics[width=1in,height=1.25in,clip,keepaspectratio]{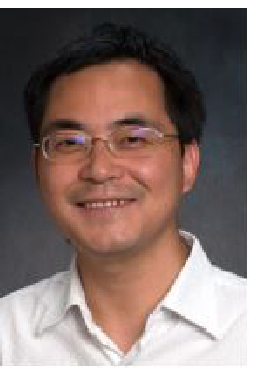}}]{J. Andrew Zhang}
(M'04-SM'11) received the B.Sc. degree from Xi'an JiaoTong University, China, in 1996, the M.Sc. degree from Nanjing University of Posts and Telecommunications, China, in 1999, and the Ph.D. degree from the Australian National University, in 2004.

Currently, Dr. Zhang is an Associate Professor in the School of Electrical and Data Engineering, University of Technology Sydney, Australia. He was a researcher with Data61, CSIRO, Australia from 2010 to 2016, the Networked Systems, NICTA, Australia from 2004 to 2010, and ZTE Corp., Nanjing, China from 1999 to 2001.  Dr. Zhang's research interests are in the area of signal processing for wireless communications and sensing. He has published more than 180 papers in leading international Journals and conference proceedings, and has won 5 best paper awards. He is a recipient of CSIRO Chairman's Medal and the Australian Engineering Innovation Award in 2012 for exceptional research achievements in multi-gigabit wireless communications.

\end{IEEEbiography}

\begin{IEEEbiography}[{\includegraphics[width=1in,height=1.25in,clip,keepaspectratio]{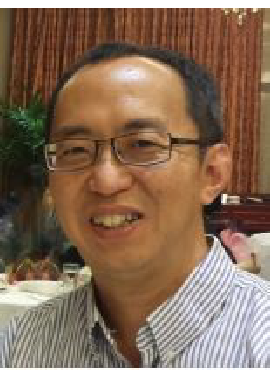}}]{Ren Ping Liu}
(M'09-SM'14) received his B.E. and M.E. degrees from Beijing University of Posts and Telecommunications, China, and the Ph.D. degree from the University of Newcastle, Australia.

He is currently a Professor and Head of Discipline of Network \& Cybersecurity at University of Technology Sydney. Professor Liu was the co-founder and CTO of Ultimo Digital Technologies Pty Ltd, developing IoT and Blockchain. Prior to that he was a Principal Scientist and Research Leader at CSIRO, where he led wireless networking research activities. He specialises in system design and modelling and has delivered networking solutions to a number of government agencies and industry customers. His research interests include wireless networking, Cybersecurity, and Blockchain.

Professor Liu was the founding chair of IEEE NSW VTS Chapter and a Senior Member of IEEE. He served as Technical Program Committee chairs and Organising Committee chairs in a number of IEEE Conferences. Prof Liu was the winner of Australian Engineering Innovation Award and CSIRO Chairman medal. He has over 200 research publications.

\end{IEEEbiography}

\end{document}